\documentclass{jocg}

\usepackage{ifthen}

\usepackage{hyperref}
\usepackage{graphicx}
\usepackage{cite}
\usepackage{todonotes}
\usepackage{verbatim}
\usepackage{amssymb,amsmath,amsthm,amsfonts}
\usepackage{mathtools}
\usepackage{placeins}
\usepackage{algorithm}
\usepackage{algpseudocode}

\makeatletter
\newcommand{\algmargin}{\the\ALG@thistlm}
\makeatother
\newlength{\whilewidth}
\algdef{SE}[parWHILE]{parWhile}{EndparWhile}[1]
{\parbox[t]{\dimexpr\linewidth-\algmargin}{%
		\hangindent\whilewidth\strut\algorithmicwhile\ #1\ \algorithmicdo\strut}}{\algorithmicend\ \algorithmicwhile}%
\algnewcommand{\parState}[1]{\State%
	\parbox[t]{\dimexpr\linewidth-\algmargin}{\strut #1\strut}}

\usepackage[mathlines]{lineno}

\title{%
	\MakeUppercase{Stable-Matching Voronoi Diagrams:\newline Combinatorial Complexity and Algorithms}%
	\thanks{This is a full version of an extended abstract presented in ICALP'18.
		Work on this paper by the first author has been supported in part by BSF Grant 2017684.
		This work was partially supported by DARPA under agreement no.~AFRL FA8750-15-2-0092 and NSF grants
		1228639, 1526631,
		1217322, 1618301, and 1616248. The views expressed are those of the authors and do not reflect the
		official policy or position of the Department of Defense
		or the U.S.~Government.}
}

\author{%
	Gill~Barequet,%
	\thanks{\affil{Technion---Israel Inst. of Technology, Haifa}, 
		\email{barequet@cs.technion.ac.il}}\,
	David~Eppstein,%
	\thanks{\affil{University of California, Irvine},
		\email{\{eppstein,goodrich,nmamano\}@uci.edu}}\,
	Michael~Goodrich,\footnotemark[3]\,
	and Nil~Mamano\footnotemark[3]
}

\usepackage{amsthm}
\theoremstyle{plain}
\newtheorem{theorem}{Theorem}
\newtheorem{lemma}[theorem]{Lemma}
\newtheorem{corollary}[theorem]{Corollary}
\theoremstyle{definition}
\newtheorem{definition}[theorem]{Definition}

\newtheorem{observation}[theorem]{Observation}
\theoremstyle{remark}

\begin{document}

\maketitle

\begin{abstract}
We study algorithms and combinatorial complexity bounds for 
\emph{stable-matching Voronoi diagrams}, where 
a set, $S$, of $n$ point sites in the plane determines a stable matching 
between the points in $\mathbb{R}^2$ and the sites in $S$ such that 
(i) the points prefer sites closer to them and sites 
prefer points closer to them, and (ii) each site has a
quota or ``appetite''
indicating the area of the set of points that can be matched to it.
Thus, a stable-matching Voronoi diagram is a solution to the
well-known post office problem with the added (realistic) constraint
that each post office has a limit on the size of its jurisdiction.
Previous work on the stable-matching Voronoi diagram provided existence and uniqueness proofs, but did not analyze its combinatorial or algorithmic complexity.
In this paper, we show that a stable-matching Voronoi diagram of $n$
point sites has $O(n^{2+\varepsilon})$ faces and edges, for any $\varepsilon>0$, and show that this bound is almost tight by giving a family of diagrams with $\Theta(n^2)$ faces and edges.
We also provide a discrete algorithm for constructing it in $O(n^3\log n+n^2f(n))$ time
in the real-RAM model of computation,
where $f(n)$ is the runtime of a geometric primitive (which we define)
that can be approximated numerically, but 
cannot, in general, be performed exactly in an algebraic model of computation. We show, however, how to compute the geometric primitive exactly for polygonal convex distance functions.
\end{abstract}

\section{Introduction}\label{smvd:sec:intro}

The \emph{Voronoi diagram} is a well-known geometric structure with a 
broad spectrum of applications in computational geometry
and other areas of Computer Science, e.g.,
see~\cite{Aurenhammer:1991,bookAurenhammer,Brandt1992, Kise1998, Meguerdichian2001, Petrek2007, Stojmenovic2006, Bhattacharya2008}.
The Voronoi diagram partitions the plane into regions.
Given a finite set $S$ of points, called \emph{sites}, each point in the plane is assigned to the region of its closest site in $S$.
Although the Voronoi diagram has been generalized in many ways, its 
standard definition specifies that each
\emph{Voronoi cell} or \emph{region} of a site $s$ is the set
$V(s)$ defined as
\begin{equation}\label{smvd:eq:vd}
\bigl\{p\in \mathbb{R}^2 \mid d(p,s)\leq d(p,s')\quad\forall s'\not=s\in S\bigr\},
\end{equation}
where $d(\cdot,\cdot)$ denotes the distance between two points.
The properties of standard Voronoi diagrams have been thoroughly 
studied (e.g., see~\cite{Aurenhammer:1991,bookAurenhammer}).
For example, it is well known that
in a standard Voronoi diagram for point sites in the plane
every Voronoi cell is a connected and convex polygon whose
boundaries lie along perpendicular bisectors of pairs of sites.

On a seemingly unrelated topic, the theory of 
\emph{stable matchings} studies how to match entities in two sets, 
each of which has its own preferences about the elements of the other set, 
in a ``stable'' manner.
It is used, for instance, to match hospitals and medical students 
starting their residencies~\cite{thematch},
as well as in on-line advertisement auctions (e.g., see~\cite{Aggarwal2009}).
It was originally formulated  by Gale and Shapley~\cite{gale62}
in the context of establishing marriages between $n$ men and $n$ women,
where each man ranks the women by preference, and, likewise, the women rank the men.
A matching between the men and women is \emph{stable}
if there is no \emph{blocking pair}, defined as 
a man and woman who prefer each other over their assigned choices.
Gale and Shapley~\cite{gale62} show that a stable solution always 
exists for any set of preferences, and they provide an algorithm that runs in
$O(n^2)$ time.

When generalized to the one-to-many case, the stable matching problem
is also known as the \emph{college admission} problem~\cite{Roth89}
and can be formulated as a matching of $n$ students to $k$ colleges,
where each student has a preference ranking of the colleges and each college
has a preference ranking of the students and a \emph{quota} 
indicating how many students it can accept.

In this paper, we are interested in studying the algorithmic and combinatorial
complexity of the diagrams that we call \emph{stable-matching Voronoi diagrams}, which
combine the notions of Voronoi diagrams and the one-to-many version of stable
matching.
These diagrams
were introduced by Hoffman, Holroyd, and Peres~\cite{hoffman2006}, who provided existence and uniqueness
proofs for such structures for potentially
countably infinite sets of sites, but did not study their
algorithmic or combinatorial complexities.
A stable-matching Voronoi diagram is defined with respect to
a set of sites in $\mathbb{R}^2$, which in this paper we restrict to 
finite sets of $n$ distinct points,
each of which has an assigned finite numerical
\emph{quota} (which is also known as its ``\emph{appetite}'') indicating the area of the region of points assigned to it.
A preference relationship is defined in terms of distance, so that 
each point $p$ in $\mathbb{R}^2$ prefers sites ordered by distance, from
closest to farthest,
and each site likewise prefers points ordered by distance.
The stable-matching Voronoi diagram, then, is a partition of the plane into regions, such that (i) each site is associated with a region of area equal to its appetite, and (ii) the assignment of points to sites is stable in
the sense that there is no blocking pair, defined as a site--point pair whose members 
prefer each other over their assigned matches.
This is formalized in Definition~\ref{smvd:def:smvd}. The regions are defined as closed sets so that boundary points lie in more than one region, analogously to Equation~\ref{smvd:eq:vd}.  
See Figure~\ref{smvd:fig:appetite}.

\begin{figure}[htb]
\centering
\includegraphics[width=.495\linewidth]{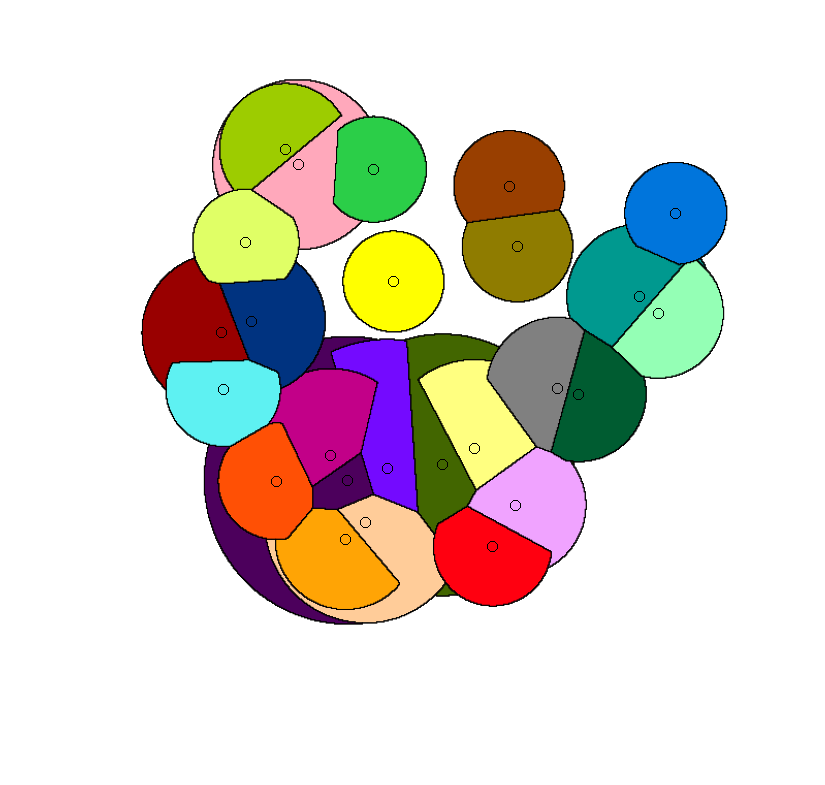}\hspace*{.1em}
\includegraphics[width=.495\linewidth]{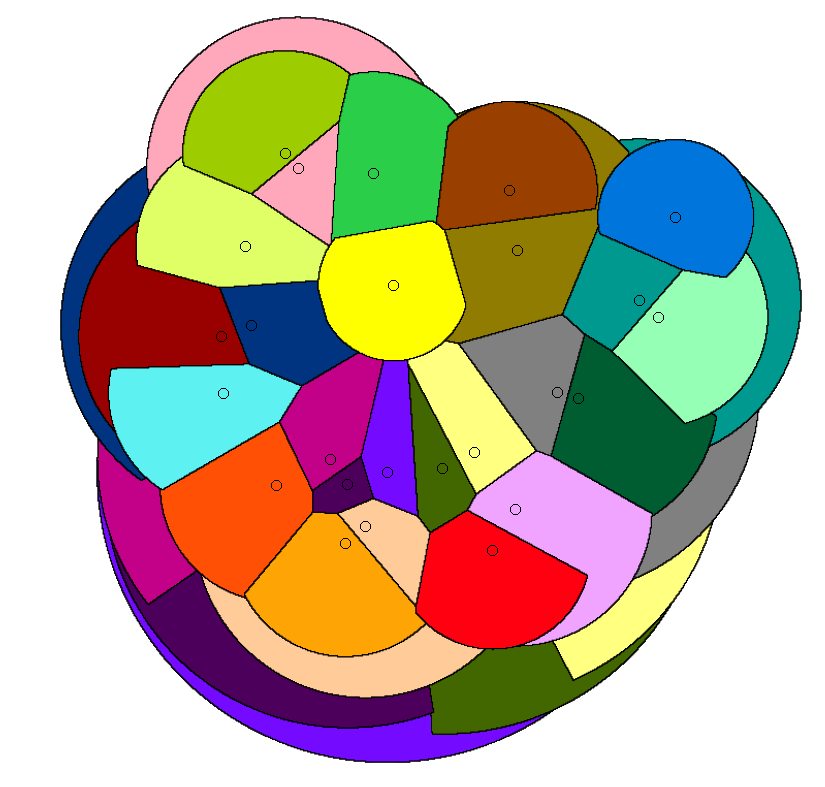}
\caption{Stable-matching Voronoi diagrams for a set of 25 point sites,
where each site in the left diagram has an 
appetite of 1 and each site in the right
diagram has an appetite of 2.
Each color corresponds to an individual cell, which is not necessarily convex or
even connected.}
\label{smvd:fig:appetite}
\end{figure}

\begin{definition}\label{smvd:def:smvd}
Given a set $S$ of $n$ points (called sites) in $\mathbb{R}^2$ and 
a numerical appetite $A(s)>0$ for each $s\in S$,
the \emph{stable-matching Voronoi diagram}
of $(S,A)$ is a subdivision of $\mathbb{R}^2$
into $n+1$ \emph{regions}, which are closed sets in $\mathbb{R}^2$.
For each site $s \in S$ there is a corresponding region $C_s$ of area $A(s)$, and there is an extra region, $C_\emptyset$, for the remaining ``unmatched'' points. The regions do not overlap except along boundaries (boundary points are included in more than one region). The regions are such that there are no blocking pairs. A blocking pair is a site $s\in S$ and a point $p\in\mathbb{R}^2$ such that (\textit{i}) $p\not\in C_s$, (\textit{ii}) $d(p,s)<\max\; \{d(p',s)\mid p'\in C_s\}$, and (\textit{iii}) $p\in C_\emptyset$ or $d(p,s)<d(p, s')$, where $s'$ is a site such that $p\in C_{s'}$.
\end{definition}

As mentioned above, 
Hoffman {\it et al.}~\cite{hoffman2006} show that, for any set of sites $S$ and appetites, the stable-matching Voronoi diagram of $S$ always exists and is unique. Technically, they consider the setting where all the sites have the same appetite, but the result applies to different appetites.
They also describe a continuous process that results in the stable-matching Voronoi diagram:
Start growing a circle from all the sites at the same time and at the same rate, 
matching the sites with all the points encountered by 
the circles that are not matched yet---when a 
site fulfills its appetite, its circle stops growing.
The process ends when all the circles have stopped growing.  

Note that this circle-growing method is analogous to a continuous version of the ``deferred acceptance'' stable matching algorithm of Gale and Shapley~\cite{gale62}. The sites correspond to the set making proposals, and $\mathbb{R}^2$ to the set accepting and rejecting proposals. The sites propose to the points in order by preference (with the growing circles), as in the deferred acceptance algorithm. The difference is that, in this setting, points receive all the proposals also in order by their own preference, so they always accept the first one and reject the rest.

Clearly, the circle-growing method can be simulated to obtain a numerical approximation of the diagram, but this would not be an effective discrete algorithm, which is one of the 
interests of the present paper.

Figure~\ref{smvd:fig:comparison} shows a side-by-side comparison of the standard and stable-matching Voronoi diagrams. Note that the standard
Voronoi diagram is stable in the same sense as the stable-matching Voronoi diagram:
by definition, every point is matched to its first choice among the sites, so there can be no blocking pairs. In fact, the standard Voronoi diagram of a set of sites can be seen as the limit of the stable-matching Voronoi diagram as all the appetites grow to infinity, in the following sense: for any point $p$ in $\mathbb{R}^2$, and for sufficiently large appetites for all the sites, $p$ will belong to the region of the same site in the standard and stable-matching Voronoi diagrams.

\begin{figure}[!hbt]
\centering
\includegraphics[width=.495\linewidth]{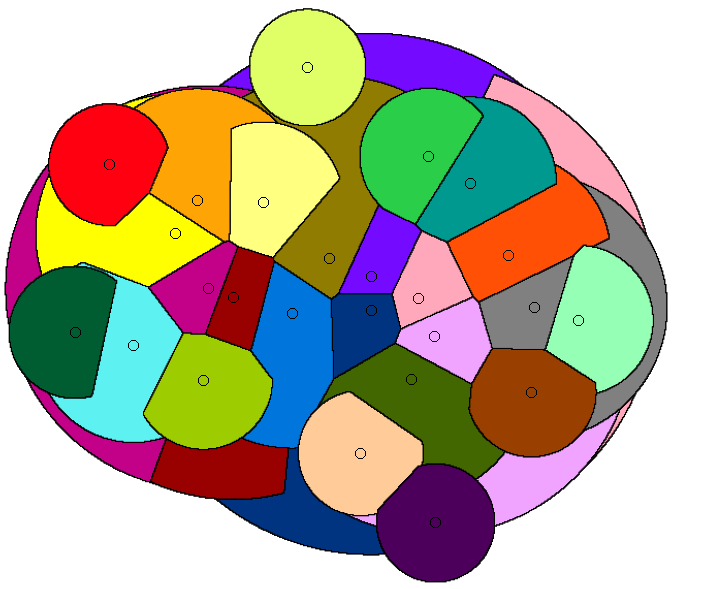}\hspace*{.1em}
\includegraphics[width=.495\linewidth]{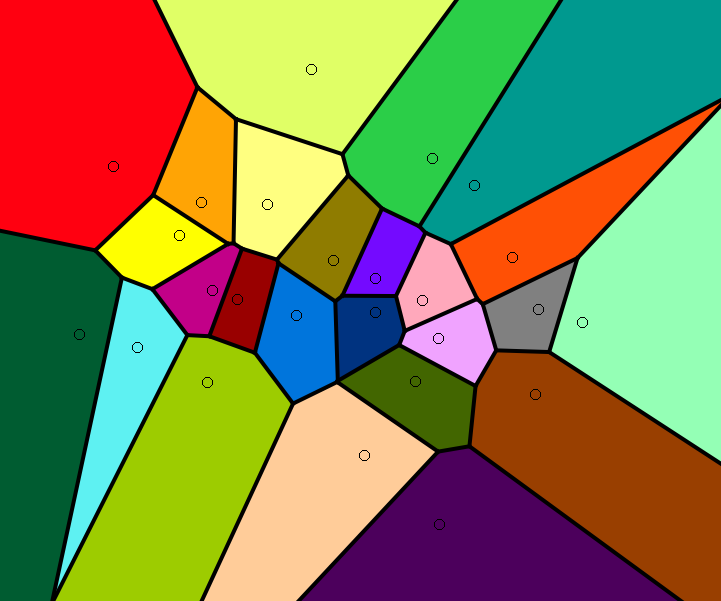}
\caption{A stable-matching Voronoi diagram (left)
and a standard Voronoi diagram (clipped to a rectangle) (right) 
for the same set of 25 sites. Each color represents a region.}
\label{smvd:fig:comparison}
\end{figure}

A standard Voronoi diagram solves the
\emph{post office} problem of assigning points 
to their closest post office~\cite{knuth1998art}.
A stable-matching Voronoi diagram adds
the real-world assumption that each post office has a limit on the size of its jurisdiction.
Such notions may also be useful for political districting,
where point sites could represent polling stations, and
appetites could represent their capacities.
In this context, the distance preferences 
for a stable-matching Voronoi diagram might determine a type of ``compactness'' 
that avoids the strange regions that are the subjects of recent court
cases involving gerrymandering. This was considered in~\cite{EPPSTEIN20172short}.
Nevertheless, depending on the appetites and 
locations of the sites, the regions of the
sites in a stable-matching Voronoi diagram are not necessarily convex or even
connected (e.g., see Figure~\ref{smvd:fig:appetite}).
Thus, we are interested in this paper in characterizing the worst-case
combinatorial complexity of such diagrams (i.e., the maximum number of faces, edges, and vertices among all diagrams with $n$ sites), as well as 
finding an efficient algorithm for constructing them.

\subsubsection*{Related Work}
There are large volumes of work on the topics of
Voronoi diagrams and stable matchings; 
hence, we refer the interested reader to 
surveys or books on the subjects
(e.g.,
see~\cite{Aurenhammer:1991,bookAurenhammer,Gusfield:1989:SMP:68392,Iwama:2008}).

A generalization of Voronoi diagram of particular interest are \emph{power diagrams}, where a weight associated to each site indicates how strongly the site draws the points in its neighborhood. Power diagrams have also been considered for political redistricting~\cite{Cohen-Addad:2018}.
Aurenhammer~\textit{et al.}~\cite{aurenhammer1998} show that, given a set of sites in a square and a quota for each site, it is always possible to find weights for the sites such that, in the power diagram induced by those weights, the area of the region of each site within the square is proportional to its prescribed quota. Thus, both stable-matching Voronoi diagrams and power diagrams are Voronoi-like diagrams that allow predetermined region sizes. Power diagrams
minimize the total squared distance between the sites and their associated points, while stable-matching Voronoi diagrams result in a stable matching.

In terms of algorithms for constructing stable-matching Voronoi diagrams, besides the mentioned continuous method by Hoffman et al.~\cite{hoffman2006}, Eppstein {\it et al.}~\cite{EPPSTEIN2017} study
the problem in a 
discrete grid setting, where both sites and points are
pixels.
Eppstein {\it et al.}~\cite{EPPSTEIN20172short} 
also consider an analogous stable-matching problem in planar graphs and road networks.
In these two previous works, the entities analogous to sites and points 
are either pixels or vertices; hence, 
they did not encounter the algorithmic 
and combinatorial challenges raised by stable-matching Voronoi diagrams for 
sites and points in the plane.
 
\subsubsection*{Our Contributions}

In Section~\ref{smvd:sec:geo}, we give a geometric interpretation of stable-matching Voronoi diagrams as the lower envelope of a set of cones, and discuss some basic properties of stable-matching Voronoi diagrams.

In Section~\ref{smvd:sec:bounds}, we give an $O(n^{2+\varepsilon})$ upper bound, for any $\varepsilon>0$, and an $\Omega(n^2)$ lower bound for the number of faces and edges of a stable-matching Voronoi diagrams in the worst case, where $n$ is the number of sites. The upper bound applies for arbitrary appetites, while the lower bound applies even in the special case where all the sites have the same appetite.

In Section~\ref{smvd:sec:algo}, we show that stable-matching Voronoi diagrams cannot be computed exactly in an algebraic model of computation. In light of this, we provide a discrete algorithm 
for constructing them that runs in $O(n^3\log n+n^2f(n))$ time,
where $f(n)$ is the runtime of a geometric primitive (which 
we define) that encapsulates this difficulty. This geometric primitive
can be approximated numerically. We also show how to compute the primitive exactly (and thus the diagram) when the distance metric is a polygonal convex distance function (Section~\ref{smvd:sec:convex}).

We assume Euclidean distance as the distance metric throughout the paper, except in Section~\ref{smvd:sec:convex}. 
In particular, the upper and lower bounds on the combinatorial complexity apply to Euclidean distance.
We conclude in Section~\ref{smvd:sec:conc}.

\section{The Geometry of Stable-Matching Voronoi Diagrams}\label{smvd:sec:geo}

As is now well known, a (2-dimensional) Voronoi diagram can be viewed as a lower envelope of cones in 3 dimensions, as follows~\cite{Fortune87}.
Suppose that the set of sites are embedded in the plane $z=0$.
That is, we map each site $s=(x_s,y_s)$ to the 3-dimensional point $(x_s,y_s,0)$.
Then, we draw one cone for each site, with the site as the vertex, and growing to $+\infty$ all with the same slope.
If we then view the cones from below, i.e., from $z=-\infty$ towards $z=+\infty$, the part of the cone of each site that we see corresponds to the Voronoi cell of the site.
This is because two such cones intersect at points that are equally distant to both vertices. As a result, the $xy$-projection of their intersection corresponds to the perpendicular bisector of the vertices, and  the boundaries of the Voronoi cells in the Voronoi diagram are determined by the perpendicular bisectors with neighboring sites.

Similarly, a stable-matching Voronoi diagram can also be viewed as the lower envelope of a set of cones. However, in this setting cones do not extending to $+\infty$. Instead, they are cut off at a finite height (which is a potentially different height for each cone, even if the associated sites have the same appetite).
This system of cones can be generated by a dynamic process that begins with cones of height zero and then grows them all at the same rate, halting the growth of each cone as soon as its area in the lower envelope reaches its appetite (see Figure~\ref{smvd:fig:3d}).
This process mimics the circle-growing method by Hoffman et al.~\cite{hoffman2006} mentioned before: if the $z$-axis is interpreted as time, the growing circles become the cones, and their lower envelope shows which circle reaches each point of the $xy$-plane first.

\begin{figure}[hbt]
\centering
\reflectbox{\includegraphics[width=.37\linewidth]{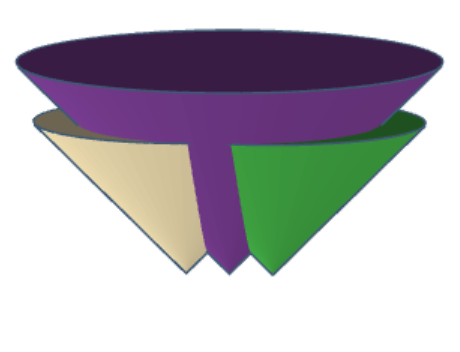}}\hspace*{3em}
\includegraphics[width=.3\linewidth]{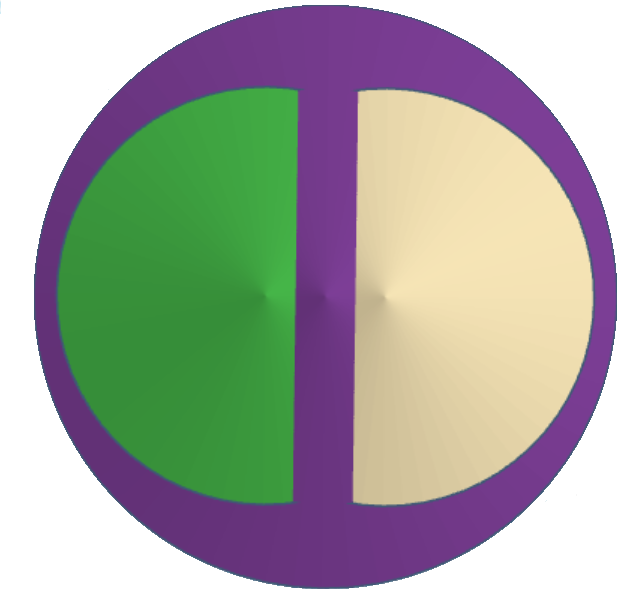}
\caption{View of a stable-matching Voronoi diagram of 3 sites 
as the lower envelope of a set of cones.}
\label{smvd:fig:3d}
\end{figure}

\noindent A stable-matching Voronoi diagram consists of three types of elements:
\begin{itemize}
\item A \emph{face} is a maximal, closed, connected subset of a stable cell.
The stable cells can be disconnected, that is, a cell can have more than one face.
There is also one or more \emph{empty faces}, 
which are maximal connected regions not assigned to any site.
One of the empty faces is the \emph{external face}, which is 
the only face with infinite area.
\item An \emph{edge} is a maximal line segment 
or circular arc on the boundary of two faces.
We call the two types of edges \emph{straight} and \emph{curved edges}, respectively.
For curved edges, we distinguish between its incident convex face 
(the one inside the circle along which the edge lies) and its incident
concave face.
\item A \emph{vertex} is a point shared by more than one edge. Generally, edges end at vertices, but curved edges may have no endpoints when they form a complete circle. This situation arises when the region of a site is isolated from other sites.

\end{itemize}
We say a set of sites with appetites is \textit{not} in general position if two curved edges of the stable-matching Voronoi diagram are tangent, i.e., touch at 
a point $p$ that is not an endpoint 
(e.g., two circles of radius 1 with centers 2 units apart).
In this special case, we consider that the curved edges are split at $p$, and that $p$ is a vertex.

In order to study the topology of the stable-matching Voronoi diagram, let the \emph{bounding disk}, $B_s$, of a site, $s$, be the smallest closed disk centered at $s$ that contains the stable cell of $s$.
The bounding disks arise in the topology of the diagram due to the following lemma:

\begin{lemma}\label{smvd:lem:boundary}
If part of the boundary between a face of site $s$ and a face of site $s'$ lies in the half-plane closer to $s$ than to $s'$, then that part of the boundary must lie along the boundary of the bounding disk $B_s$, and the convex face must belong to $s$.
\end{lemma}

\begin{proof}
The boundary between the faces of $s$ and $s'$ cannot lie outside of $B_s$, by definition of the bounding disk.
If the boundary is in the half-plane closer to $s$, then it also cannot be in the interior of $B_s$, because then there would exist a point $p$ inside $B_s$ and in the half-plane closer to $s$, but matched to $s'$ (see Figure~\ref{smvd:fig:boundary}).
In such a situation, $s$ and $p$ would be a blocking pair: $s$ prefers $p$ to the point(s) matched to it along $B_s$, and $p$ prefers $s$ to $s'$.
\end{proof}

\begin{figure}
\centering
\includegraphics[width=.3\linewidth]{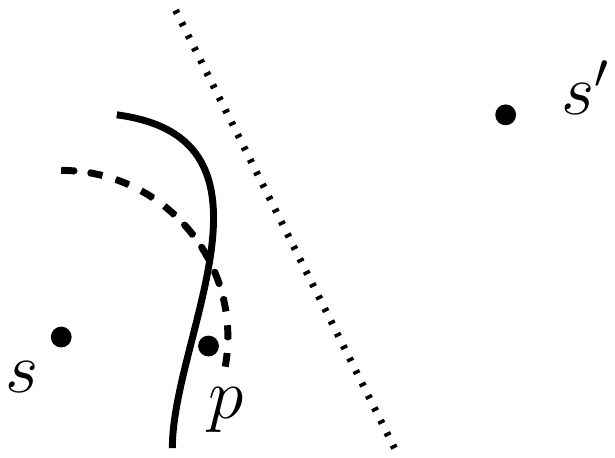}
\caption{Illustration of the setting in the proof of Lemma~\ref{smvd:lem:boundary}.
It shows the perpendicular bisector of two sites $s$ and $s'$ (dotted line), the boundary of the bounding disk, $B_s$, of $s$ (dashed circular arc), and a hypothetical boundary between the faces of sites $s$ and $s'$ (solid curve).
In this setting, $s$ and $p$ would be a blocking pair.}
\label{smvd:fig:boundary}
\end{figure}

\begin{lemma}\label{smvd:lem:bounding}
The union of non-empty faces of the diagram is the union of the bounding disks of all the sites.
\end{lemma}

\begin{proof}
For any site $s$, all the points inside the bounding disk of $s$ must be matched.
Other\-wise, there would be a point, say, $p$, not matched to anyone but closer to $s$ than points actually matched to $s$ (along the boundary of $B_s$), which would be unstable, as $p$ and $s$ would be a blocking pair.
Moreover, points outside of all the bounding disks cannot be matched to anyone, by definition of the bounding disks.
\end{proof}

\begin{lemma}[Characterization of edges]\label{smvd:lem:edges}
$ $

\begin{enumerate}
\item A straight edge separating faces of sites $s$ and $s'$ can only lie along the perpendicular bisector of $s$ and $s'$.
\item A curved edge whose convex face belongs to site $s$ lies along the boundary of $B_s$.
Moreover, if the concave face belongs to a site $s'$, the edge must be contained in the half-plane closer to $s$ than $s'$.
\item Empty faces can only be concave faces of curved edges.
\end{enumerate}
\end{lemma}

\begin{proof}
Claims~(1) and~(2) are consequences of Lemma~\ref{smvd:lem:boundary}, and Claim~(3) is a consequence of Lemma~\ref{smvd:lem:bounding}.
\end{proof}


\section{Combinatorial Complexity}\label{smvd:sec:bounds}

\subsection{Upper Bound on the Number of Faces}

As mentioned in Section~\ref{smvd:sec:geo}, a stable-matching Voronoi diagram can be viewed as the lower envelope of a set of cones.
Sharir and Agarwal~\cite{SA95} provide results that characterize the combinatorial complexity of the lower envelope of certain sets of functions, including cones.

Formally, the \emph{lower envelope} (also called \emph{minimization diagram}) of a set of bivariate continuous functions $F=\{f_1(x,y),\ldots,f_n(x,y)\}$ is the function
$$E_F(x,y)=\min_{1\leq i\leq n} f_i(x,y),$$
where ties are broken arbitrarily.
The lower envelope of $F$ subdivides the plane into maximal connected regions such that $E_F$ is attained by a single function $f_i$ (or by no function at all).
The \emph{combinatorial complexity} of the lower envelope $E_F$, denoted $K(F)$, is the number of maximal connected regions of $E_F$.
To prove our upper bound, we use the following result:

\begin{lemma}[Sharir and Agarwal~\cite{SA95}, page 191]\label{smvd:lem:envelope}
The combinatorial complexity $K(F)$ of the lower envelope of a collection $F$ of $n$ (partially defined) bivariate functions that satisfy the assumptions below is $O(n^{2+\varepsilon})$, for any $\varepsilon>0$.\footnote{The theorem, as stated in the book (Theorem 7.7), includes some additional assumptions, but the book then shows that they are not essential.}
\begin{itemize}
\item Each $f_i\in F$ is a portion of an algebraic surface of the form $P_i(x,y)$, for some polynomial $P_i$ of constant maximum degree.
\item The vertical projection of each $f_i\in F$ onto the $xy$-plane is a planar region bounded by a constant number of algebraic arcs of constant maximum degree.
\end{itemize} 
\end{lemma}

\begin{corollary}\label{smvd:cor:upper}
A stable-matching Voronoi diagram for $n$ sites has $O(n^{2+\varepsilon})$ faces, for any $\varepsilon>0$.
\end{corollary}
\begin{proof}
It is clear that the finite, ``upside-down'' cones whose lower envelope forms the stable-matching Voronoi diagram of a set of sites satisfy the above assumptions. In particular, their projection onto the $xy$-plane are disks. 
Note that the bound still applies if we include the empty faces, as Lemma~\ref{smvd:lem:envelope} still holds if we add an extra bivariate function $f_{n+1}(x,y)=z^*$, where $z^*$ is any value higher than the height of any cone (i.e., $f_{n+1}$ is a plane that ``hovers'' over the cones). Such a function would have a face in the lower envelope for each empty face in the stable-matching Voronoi diagram.
\end{proof}

\subsection{Upper bound on the Number of Edges and Vertices}\label{smvd:app:boundedges}
Euler's formula relates the number of faces in a planar graph with the number of vertices and edges. By viewing the stable-matching Voronoi diagram as a graph, we can use Euler's formula to prove that the $O(n^{2+\varepsilon})$ upper bound also applies to the number of edges and vertices. In order to do so, we will need to show that the average degree is more than two, which is the purpose of the following lemmas.

In this section, we assume that sites are in general position (as defined in Section~\ref{smvd:sec:geo}).
However, note that non-general-position
constructions cannot yield the worst-case complexity.
This is because if two curved edges coincide exactly at a point that is not an endpoint, we can perturb slightly the site locations to move them a little closer, which creates a new vertex and edge.
For the same reason, we also assume that no vertex has degree four or more, which requires four or more sites to lie on the same circle (as in the standard Voronoi Diagram).

\begin{lemma}\label{smvd:lem:edgesequences}
	The following sequences of consecutive edges along the boundary of two faces cannot happen: \textbf{1.} Straight--straight. \textbf{2.} Curved--curved. \textbf{3.} Straight--curved--straight.
\end{lemma}

\begin{proof}
	$ $
	\begin{enumerate}
		\item Straight edges separating two faces of sites $s$ and $s'$ are constrained to lie along the perpendicular bisector of $s$ and $s'$ (Lemma~\ref{smvd:lem:edges}).
		Therefore, two consecutive straight edges would not be maximal.
		\item Curved edges separating a convex face of a site $s$ are constrained to lie along the boundary of the bounding disk of $s$ (Lemma~\ref{smvd:lem:edges}).
		Thus, two consecutive curved edges would not be maximal (under the assumption of general position).
		\item In such a case, not both straight edges could lie along the perpendicular bisector.\qedhere
	\end{enumerate}
\end{proof}
Incidentally, \emph{curved--straight--curved} sequences \emph{can} happen, and can be seen in Figure~\ref{smvd:fig:comparison}.

\begin{lemma}\label{smvd:lem:degreetwo}
	A vertex with degree two cannot be adjacent to two vertices with degree two.
\end{lemma}

\begin{proof}
	A vertex with degree two connects two edges separating the same two faces.
	If there were a node with degree two adjacent to two other nodes with degree two, we would either have four consecutive edges separating the same two faces or a triangular face inside another face.
	However, neither case could avoid the sequences of edges given in Lemma~\ref{smvd:lem:edgesequences}.
\end{proof}

\begin{lemma}~\label{smvd:lem:avgdegree}
	The average degree is at least $2.25$.
\end{lemma}

\begin{proof}
	Note that all vertices have degree at least $2$, as they are the endpoints of edges, and every edge has different faces on each side. Also recall the assumption that there are no nodes with degree more than three, as this cannot yield a worst-case number of vertices nor edges.
	
	Thus, all nodes have degree two or three. Let $n$ be the number of 2-degree vertices, and $k$ the number of 3-degree vertices.
	The average degree is $(2n+3k)/(n+k)=2+k/(n+k)$. Thus, we need to show that $k/(n+k)\geq 1/4$, or, rearranging, that $k\geq n/3$.
	
	By Lemma~\ref{smvd:lem:degreetwo}, a vertex with degree two cannot be adjacent to two vertices with degree two.
	Among the $n$ 2-degree vertices, say $m_1$ are connected with another 2-degree node, while the remaining $m_2$ are adjacent only to 3-degree nodes. Then, there are $m_1+2m_2=n+m_2$ edges connecting 2-degree nodes with 3-degree nodes.
	This means that $k\geq (n+m_2)/3$, completing the proof.	
\end{proof}

\begin{lemma}\label{smvd:lem:upperedges}
Let $V,E$ and $F$ be the number of vertices, edges, and faces of the stable-matching Voronoi diagram of a set of sites $S$. Then, $V\leq 8F-16$ and $E\leq 9F-18$.
\end{lemma}

\begin{proof}
For this proof, suppose that there are no curved edges that form a full circle. Note that the presence of such edges can only reduce the number of vertices and edges, as for each such edge there is a site with a single edge and no vertices.

Without such edges, the vertices and edges of the stable-matching Voronoi diagram form a planar graph, and $V,E,F$ are the number of vertices, edges, and faces of this graph, respectively. Moreover, let $C$ be the number of connected components.
Due to Euler's formula for planar graphs, we have $F=E-V+C+1$, and thus $F\geq E-V+2$.
Moreover, by Lemma~\ref{smvd:lem:avgdegree}, the sum of degrees is at least $2.25V$, so $2E\geq 2.25V$.
Combining the two relations above, we have $V\leq 8F-16$ and $E\leq 9F-18$.
\end{proof}

We conclude by stating the main theorem of this section, which is a combination of Corollary~\ref{smvd:cor:upper} and Lemma~\ref{smvd:lem:upperedges}:
\begin{theorem}\label{smvd:the:upbound}
A stable-matching Voronoi diagram for $n$ point sites 
has $O(n^{2+\varepsilon})$ faces, vertices, and edges, for any $\varepsilon>0$.
\end{theorem}

\subsection{Lower Bound}\label{smvd:sec:lower}

We show a quadratic lower bound on the number of faces in the worst case by constructing an infinite family of instances with $\Omega(n^2)$ faces. To start, we give such a family of instances where sites have arbitrary appetites. This introduces the technique behind our second, more intricate construction, which only uses sites with appetite $1$. This shows that the $\Omega(n^2)$ lower bound holds even in this restricted case where all the sites have the same appetite. The lower bound extends trivially to vertices and edges as well. 

\begin{lemma}\label{smvd:lem:lowerPre}
	A stable-matching Voronoi diagram for $n$ point sites
	has $\Omega(n^2)$ faces, edges, and vertices in the worst case.
\end{lemma}

\begin{proof}
	Consider the setting in Figure~\ref{smvd:fig:lowerPre}.
	Assume $n$ is even. We divide the sites into two sets, $X$ and $Y$, of size $m=n/2$ each.
	The sites in $X$ are arranged vertically, spaced evenly, and spanning a total height of $2$. Note that the \textit{standard} Voronoi diagram of the sites in $X$ alone consists of infinite horizontal strips. The top and bottom sites have strips extending vertically indefinitely, while the rest have thin strips of height $2/(m-1)$.
	
	The sites in $Y$ are aligned vertically with the center of the strips. Half of the sites in $Y$ lie on each side of the sites in $X$. The sites in $Y$ have appetite $\pi$, so their ``ideal'' stable cell is a disk of radius $1$ around them. They are spaced evenly at a distance of at least 2 (e.g., $2.1$) of each other and of the first $m$ sites, so that each site in $Y$ is the first choice of all the points within distance $1$ of it.
	
	Now, consider the resulting stable-matching Voronoi diagram when the sites of $X$ have large ($\gg m^2$) and equal appetites. 
	To visualize it, consider the circle-growing method from~\cite{hoffman2006} described in Section~\ref{smvd:sec:intro}, where a circle starts growing from each site at the same time and rate, and any unassigned point reached by a circle is assigned to the corresponding site.
	
	The sites in $Y$ are allowed to grow without interference with any other site until they fulfill their appetite and freeze. Their region is thus a disk with diameter $2$, which spans all the thin strips (Figure~\ref{smvd:fig:lowerPre}).
	The sites in $X$ start growing their region as a disk, which quickly reach the disks of the sites above and below. Then, the regions are restricted to keep growing along the horizontal strips. 
	The sites in $X$ keep growing and eventually reach a region already assigned to a site in $Y$ (which already fulfilled their appetite and stopped growing by this time). They continue growing along the strips past the regions already assigned to sites in $Y$. Eventually, they also freeze when they fulfill their appetite. The top and bottom sites are the first to fulfill their appetite, since they are not restricted to grow along thin strips, but we are not interested in the topology of the diagram beyond the thin strips. The only thing we need for our construction is that the appetite of the sites in $X$ is large enough so that their stable cells reach past the stable cells of the furthest sites in $Y$ along the strips.
	
	Informally, the regions of the sites in $Y$ ``cut'' the thin strips of the sites in $X$. Each site in $Y$ creates $m-2$ additional faces (the top and bottom sites in $X$ do not have thin strips), and hence the number of faces is at least $m(m-2)=\Omega(n^2)$. 
\end{proof}

\begin{figure}
	\centering
	\includegraphics[width=\linewidth]{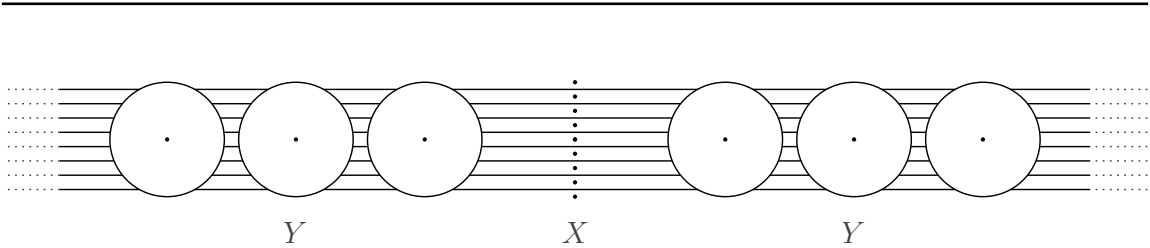}
	\caption{Lower bound construction for Lemma~\ref{smvd:lem:lowerPre}.}
	\label{smvd:fig:lowerPre}
\end{figure}

In the proof of Lemma~\ref{smvd:lem:lowerPre}, sites in $X$ and $Y$ have different roles. Sites in $X$ create a linear number of long, thin faces which can all be cut by a single disk. This is repeated a linear number of times, once for each site in $Y$, yielding quadratic complexity. However, this construction relies on the sites in $X$ having larger appetites than the sites in $Y$. Next, we consider the case where all the sites have appetite one.
The proof will follow the same idea, but now the thin and long strips will be circular strips. 
Lemma~\ref{smvd:lem:annulusintersection} is an auxiliary result used in the proof.

\begin{lemma}\label{smvd:lem:annulusintersection}
	Let $A$ be an annulus of width $\varepsilon>0$, and $D$ a disk centered outside the outer circle of $A$, with radius smaller than the inner radius of $A$, and tangent to the inner circle of $A$ (Figure~\ref{smvd:fig:annulusintersection}). Then,
	$$\lim\limits_{\varepsilon\rightarrow 0}\frac{area(A\cap D)}{area(A)}=0$$
\end{lemma}
\begin{figure}
	\centering
	\includegraphics[width=.43\linewidth]{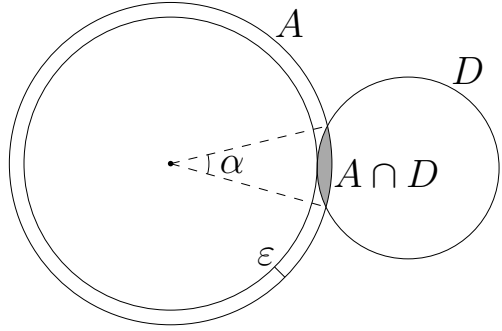}
	\caption{Setting in Lemma~\ref{smvd:lem:annulusintersection}.}
	\label{smvd:fig:annulusintersection}
\end{figure}

\begin{proof}
	Consider the smallest circular sector $S$ of $A$ that contains the asymmetric lens $A\cap D$ (the sector determined by angle $\alpha$ in Figure~\ref{smvd:fig:annulusintersection}). Since $A\cap D$ is contained in $S$, to prove the lemma it suffices to show that
	$\lim\limits_{\varepsilon\rightarrow 0}\frac{area(S)}{area(A)}=0$. Note that $\frac{area(S)}{area(A)}$ is precisely $\frac{\alpha}{2\pi}$, and it is clear that $\lim\limits_{\varepsilon\rightarrow 0}\frac{\alpha}{2\pi}=0$. 
\end{proof}

\begin{theorem}\label{smvd:thm:lower}
	A stable-matching Voronoi diagram for $n$ point sites
	has $\Omega(n^2)$ faces, edges, and vertices in the worst case, even when all the regions are restricted to have the same appetite.
\end{theorem}

\begin{proof}
Assume $n$ is a multiple of $4$. We divide the sites into two sets, $X$ and $Y$, of size $m=n/2$ each.
	
Let $\varepsilon_1,\varepsilon_2$ be two parameters with positive values that may depend on $m$. It will be useful to think of them as very small, since we will argue that the construction works for sufficiently small values of $\varepsilon_1$ and $\varepsilon_2$. Specific values for $\varepsilon_1$ and $\varepsilon_2$ are hard to express analytically but unimportant as long as they are small enough.
	
The $m$ sites in $X$, $s_1,\ldots,s_m$, lie, in this order, along a circle of radius $\varepsilon_1$.
They are almost evenly spaced around the circle, except that the angle between $s_1$ and $s_m$ is slightly larger than the others: the angle between $s_1$ and $s_m$ is increased by $\varepsilon_2$, and the angles between the rest of pairs of consecutive sites are reduced so that they are all equal (Figure~\ref{smvd:fig:lowerHalf}, Left).

The \textit{standard} Voronoi diagram of the sites in $X$ consists of infinite angular regions, with those of $s_1$ and $s_m$ slightly wider than those of the remaining sites.
Consider the circle-growing method applied to the sites of $X$ alone.
Initially, the regions are constrained to grow in the corresponding angular region in the standard Voronoi region. Since $s_1$ and $s_m$ have wider angles, they fill their appetite slightly before the rest, which all grow at the same rate. How much earlier depends on $\varepsilon_2$. Once $s_1$ and $s_m$ fulfill their appetite and stop growing, their angular regions become ``available'' to the other sites. The circles of $s_2$ and $s_{m-1}$ are the closest to the angular regions of $s_1$ and $s_m$, respectively, and thus start covering it to fulfill their appetite. In turn, this results in $s_2$ and $s_{m-1}$ fulfilling their appetite and freezing their circles earlier than the remaining sites. Their respective neighbors, $s_3$ and $s_{m-2}$, have the next closest circles to the angular regions of the sites that already stopped growing, and thus they use it to fill their appetite. This creates a cascading effect starting with $s_1$ and $s_m$ where the region of each site consists of a wedge that ends in a thin circular strip that ``wraps around'' the regions of the prior sites (Figure~\ref{smvd:fig:lowerHalf}, Right).

\begin{figure}
	\centering
	\includegraphics[width=.99\linewidth]{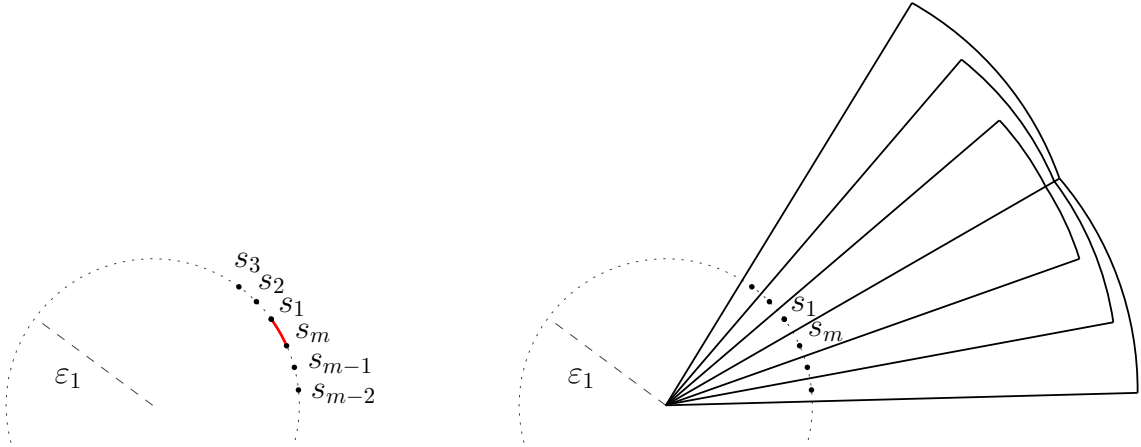}
	\caption{Left: Configuration of the sites in $X$ in the proof of Theorem~\ref{smvd:thm:lower}. The arc between $s_1$ and $s_m$, shown in red, is slightly wider than the rest. Sites $s_2,s_3,s_{m-1},s_{m-2}$ are shown. The remaining sites around the circle are omitted for clarity. Right: the stable cells of the aforementioned sites. The figures are not to scale, as in the actual construction $\varepsilon_1$ and $\varepsilon_2$ need to be much smaller, but even here we can appreciate the ``wrapping around'' effect.}
	\label{smvd:fig:lowerHalf}
\end{figure}

As $\varepsilon_2$ approaches zero, the unfulfilled appetite of the sites other than $s_1$ and $s_m$ at the time $s_1$ and $s_m$ fill their appetite becomes arbitrarily small. 
This results in arbitrarily thin circular strips. Note, however, that the circular arcs bounding each strip are not exactly concentric, as each one is centered at a different site. Thus, depending on $\varepsilon_1$, the strips might not wrap around all the way to the regions of $s_1$ and $s_m$. However, as $\varepsilon_1$ approaches zero, the sites get closer to each other, and thus their circular arcs become arbitrarily close to being concentric. It follows that if $\varepsilon_1$ is small enough (relative to $\varepsilon_2$), the circular strip of each site will wrap around all the way to the angular region of $s_1$ and $s_m$.
This concludes the first half of the construction, where we have a linear number of arbitrarily
thin, long strips. 

Let $A$ be the annulus of minimum width centered at the center of the circle of the sites in $X$ and containing all the circular strips.
The sites in $Y$ lie evenly spaced along a circle concentric with $A$.
The radius of the circle is such that the regions of the sites in $Y$ are tangent to the inner circle of $A$, as in Figure~\ref{smvd:fig:lowerHalf2}.  

\begin{figure}
	\centering
	\includegraphics[width=.45\linewidth]{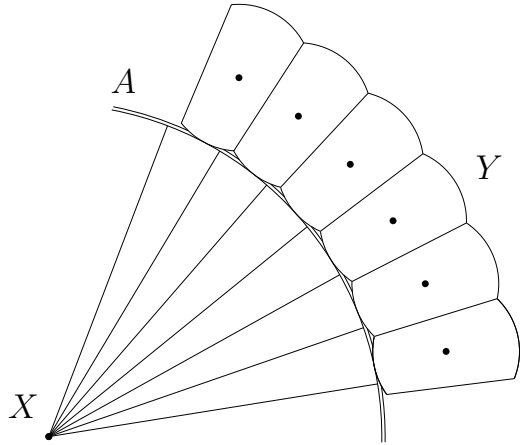}
	\caption{Configuration in the proof of Theorem~\ref{smvd:thm:lower}. For clarity, only the regions of $6$ sites in $X$ and 6 sites in $Y$ are shown. The strips inside $A$ are also omitted. The figure is not to scale, as in the actual construction the annulus $A$ needs to be much thinner.}
	\label{smvd:fig:lowerHalf2}
\end{figure}

Since the wedges of the sites in $X$ are very
thin, the sites in $Y$ are closer to $A$ than the sites in $X$. Thus, the presence of $X$ does not affect the stable cells of the sites in $Y$. Each stable cell of a site in $Y$ is the intersection of a disk and a wedge of angle $2\pi/m$, with a total area of $1$ (Figure~\ref{smvd:fig:lowerHalf2}). The important aspect is how the presence of the stable cells of the sites in $Y$ affects the stable cells of the sites in $X$.
Some of the area of $A$ that would be assigned to sites in $X$ is now assigned to sites in $Y$. Thus, the sites in $X$ need to grow further to make up for the lost appetite.
However, recall that $A$ can be arbitrarily thin. Hence, by Lemma~\ref{smvd:lem:annulusintersection}, the fraction of the area of $A$ ``eaten'' by sites in $Y$ can be arbitrarily close to zero. As this fraction tends to zero, the distance that the sites in $X$ need to reach further to fulfill the lost appetite also tends to zero. Thus, if $A$ is sufficiently thin, the distance that the regions of $s_1$ and $s_m$ reach further is so small that the strips of $s_2$ and $s_{m-1}$ still wrap around the regions of $s_1$ and $s_m$, respectively, to fulfill their appetite. Similarly, the strips of $s_3$ and $s_{m-2}$ still wrap around the regions of the prior sites, and so on. Thus, if $A$ is sufficiently thin, the strips of all the sites in $X$ still wrap around to the regions of $s_1$ or $s_m$.

In this setting, half of the strips are at least as long as a quarter of the circle, and each of those gets broken into $\Theta(m)$ faces by the regions of the sites in $Y$.
Therefore, the circular strips are collectively broken into 
a total of $\Theta(m^2)=\Theta(n^2)$ faces.
\end{proof}

\section{Algorithm}\label{smvd:sec:algo}

In general, a stable-matching Voronoi diagram cannot be computed in an algebraic model of computation, as it requires computing transcendental functions such as trigonometric functions.

\begin{observation}\label{smvd:obs:transcentental}
	For infinitely-many sets of sites in general position and
	with algebraic coordinates, the radii of some of the sites' bounding
	disks cannot be computed exactly in an algebraic model of computation.
\end{observation}

\begin{proof}
	Consider a set with only two sites, $s_1$ and $s_2$, with appetite~1 and aligned horizontally at distance~$2b$
	from each other. By symmetry, the two bounding disks will have the same
	radius~$r$. Assume that $b<\sqrt{1/\pi}$, so that the stable cells of $s_1$ and $s_2$ share a vertical edge. Consider the rectangular triangle with one vertex at $s_1$, another at the midpoint between $s_1$ and $s_2$, and the last at the top of the shared vertical edge (see Figure~\ref{smvd:fig:trigo}). Let $\alpha$ be the angle of the triangle at the vertex at $s_1$, and $a$ the length of the opposite side. The problem is, then, to determine the value of~$r$ which
	satisfies 
	\[
	\pi r^2 \left(1-\frac{2\alpha}{2\pi}\right) + 2\cdot \frac{ab}{2} = 1.
	\]
	Using the equalities~$\sin \alpha = a/r$ and~$\cos \alpha = b/r$, we obtain
	\[
	r^2 (\pi - \cos^{-1} \frac{b}{r}) + br \sin (\cos^{-1} \frac{b}{r}) = 1,
	\]
	that is,~$r$ is the solution of the equation
	\[
	r^2 (\pi - \cos^{-1} \frac{b}{r}) + b \sqrt{r^2-b^2} = 1,
	\]
	which cannot be solved in an algebraic model of computation because $\cos^{-1}$ is a transcendental (i.e., non-algebraic) function.
	Such a construction appears in infinitely-many sets of points, implying the claim.
\end{proof}
\begin{figure}[htb]
	\centering
	\includegraphics[width=.475\linewidth]{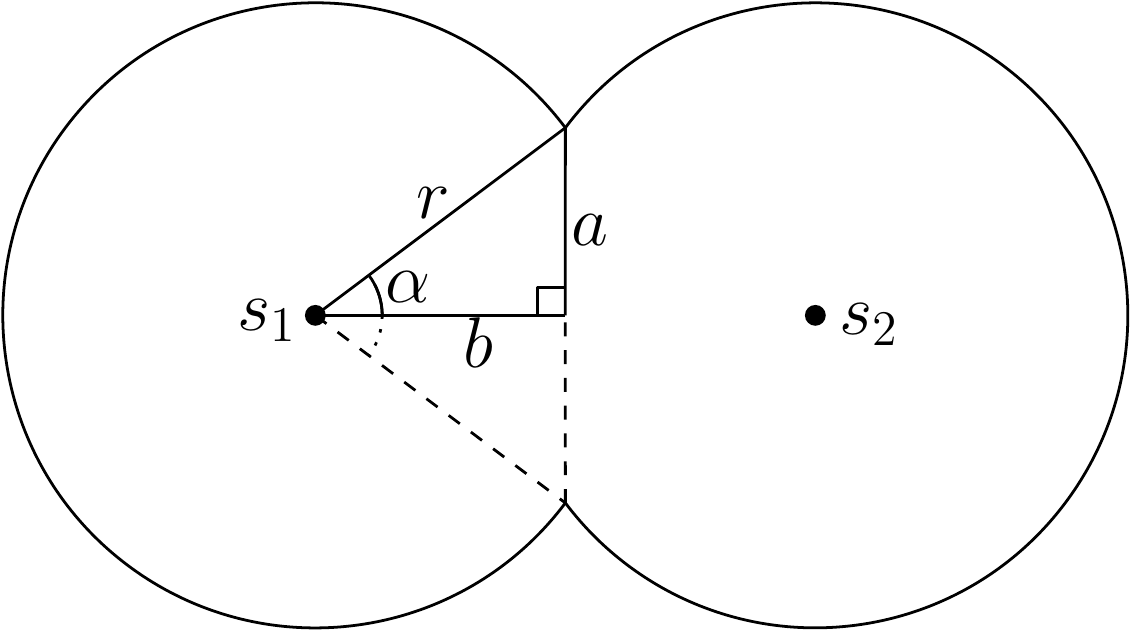}
	\caption{Setting in the proof of Observation~\ref{smvd:obs:transcentental}.}
	\label{smvd:fig:trigo}
\end{figure}

Thus, in order to describe an exact and discrete algorithm, we rely on a geometric primitive. This primitive, which we define, encapsulates the problematic computations, and can be approximated numerically to arbitrary precision. In Section~\ref{smvd:sec:convex}, we show how to compute the geometric primitive exactly for polygonal convex distance functions.

\paragraph{Preliminaries.}
Let us introduce the notation used in this section.
The algorithm deals with multiple diagrams. 
In this context, a \emph{diagram} is a subdivision of $\mathbb{R}^2$ into \emph{regions}. Each region is a set of one or more faces bounded by straight and circular edges. The regions do not overlap except along boundaries (boundary points are included in more than one region).
Each region is assigned to a unique site, but not all sites necessarily have a region. There is also an ``unassigned'' region consisting of the remaining faces. The \emph{domain} of a diagram is the subset of points of $\mathbb{R}^2$ in any of its assigned regions.
If $D$ is a diagram, $D(s)$ denotes the region of site $s$, which might be empty. If $D$ and $D'$ are diagrams and the domain of $D$ is a subset of the domain of $D'$, we say that $D$ and $D'$ are \emph{coherent} if, for every site $s$, $D(s)\subseteq D'(s)$. The data structures used to represent diagrams are discussed later.

Recall that we are given a set $S$ of $n$ sites, each with its own appetite $A(s)$.
The goal is to compute the (unique) stable-matching diagram of $S$ for those appetites, denoted by $D^*$.
For a site $s$, let $B^*(s)$ be the bounding disk of $D^*(s)$, and $r^*(s)$ the radius of $B^*(s)$ (the $^*$ superscript is used for notation relating to the sought solution).
Recall that the union of all the bounding disks $B^*(s)$ equals the domain of $D^*$ (Lemma~\ref{smvd:lem:bounding}), and that the bounding disks may not be disjoint.

We call an ordering $s_1,\ldots,s_n$ of the sites of $S$ \textit{proper} if the sites are sorted by increasing radius of their bounding disks, breaking ties arbitrarily. That is, $i<j$ implies $r^*(s_i)\leq r^*(s_j)$. Such an ordering is initially unknown, but it is discovered in the course of the algorithm.
Given a proper ordering, for $i=1,\ldots,n$,
let $B_{1..i}=\{B^*(s_1), \ldots, B^*(s_i)\}$ denote the set of bounding disks of the first $i$ sites, and $\cup B_{1..i}=B^*(s_1) \cup \cdots \cup B^*(s_i)$ the union of those disks. Let $\hat{B}(s_i)=B^*(s_i)\setminus \cup B_{1..i-1}$ be the part of $B^*(s_i)$ that is not inside a prior bounding disk in the ordering.
Let $S_{i..n}=\{s_i,\ldots,s_n\}$, and $V_{i..n}$ be the \emph{standard} Voronoi diagram of $S_{i..n}$.
Finally, let $\hat{V}_{i..n}$ be $V_{i..n}$ restricted to the region $\hat{B}(s_i)$.
This notation is illustrated in Figure~\ref{smvd:fig:notation}.
\begin{figure}[h!]
	\centering
	\includegraphics[width=.55\linewidth]{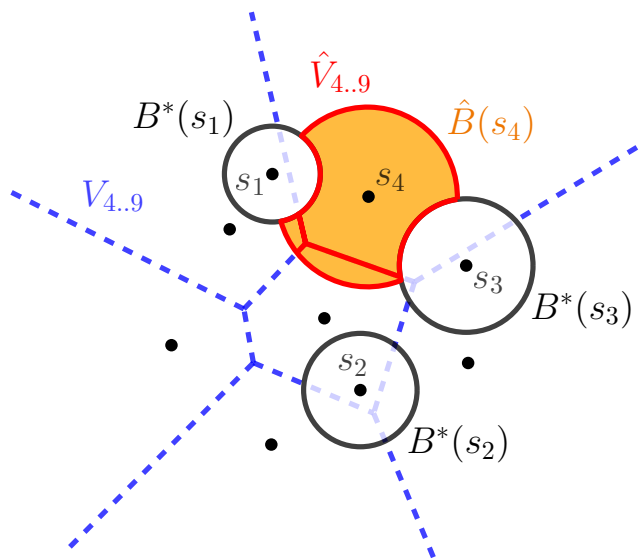}
	\caption{Notation used in the algorithm. The disks in $B_{1..3}$ are shown in black, the edges of $V_{4..9}$ are shown dashed in blue, and the interior of $\hat{B}(s_4)$ is shown in orange. The edges of $\hat{V}_{4..9}$ are overlaid on top of everything with red lines. Note that $\hat{V}_{4..9}$ is a diagram with three assigned regions, the largest assigned to $s_4$ and the others to unlabeled sites.}
	\label{smvd:fig:notation}
\end{figure}

\paragraph{Incremental construction.}
The algorithm constructs a sequence of diagrams, $D_0,\dotsc,D_n$. 
The starting diagram, $D_0,$ has an empty domain. We expand it incrementally until $D_n=D^*$. The diagrams are constructed in a greedy fashion: every $D_i$ is coherent with $D^*$. Thus, once a subregion of the plane is assigned in $D_i$ to some site, that assignment is definitive and remains part of $D_{i+1},\ldots, D_n$. 

We construct $D^*$ one bounding disk at a time, ordered according to a proper ordering  $s_1,\ldots,s_n$  (we address how to find this ordering later): the domain of each $D_i$ is $\cup B_{1..i}$.\footnote{An intuitive alternative approach is to construct $D^*$ one stable cell at a time. This is also possible, but the advantage of constructing it by bounding disks is that the topology of the intermediate diagrams $D_i$ is simpler, as it can be described as a union of disks, whereas stable cells have complex (and even disjoint) shapes. The simpler topology makes the geometric operations we do on these diagrams easier, in particular the geometric primitive from Definition~\ref{smvd:def:prim}.} Thus, $D_i$ can be constructed from $D_{i-1}$ by assigning $\hat{B}(s_i)$ (the $ \hat{\mbox{ } } $ mark is used for notation relating to the region added to $D_i$ at iteration $i$).

Since the boundaries of the bounding disks do not necessarily align with the edges of $D^*$, $D_i$ may contain a face of $D^*$ only partially. This can be seen in Figure~\ref{smvd:fig:newalgo}, which illustrates the first few steps of the incremental construction.
Further figures can be seen in Appendix~\ref{smvd:app:steps}.

At iteration $i$, we assign $\hat{B}(s_i)$ as follows. From $\hat{B}(s_i)$ and the standard Voronoi of the remaining sites, $V_{i..n}$, we compute the diagram $\hat{V}_{i..n}$. We then construct $D_i$ as the combination of $D_{i-1}$ and $\hat{V}_{i..n}$. That is, for each site $s$, $D_i(s)=D_{i-1}(s)\cup \hat{V}_{i..n}(s)$.
We first show that this assignment is correct.

\begin{figure}[t!] \centering
	\begin{minipage}[b]{0.48\linewidth}
		\fbox{\includegraphics[trim={0 4.7cm 0 2cm},clip,width=0.973\linewidth]{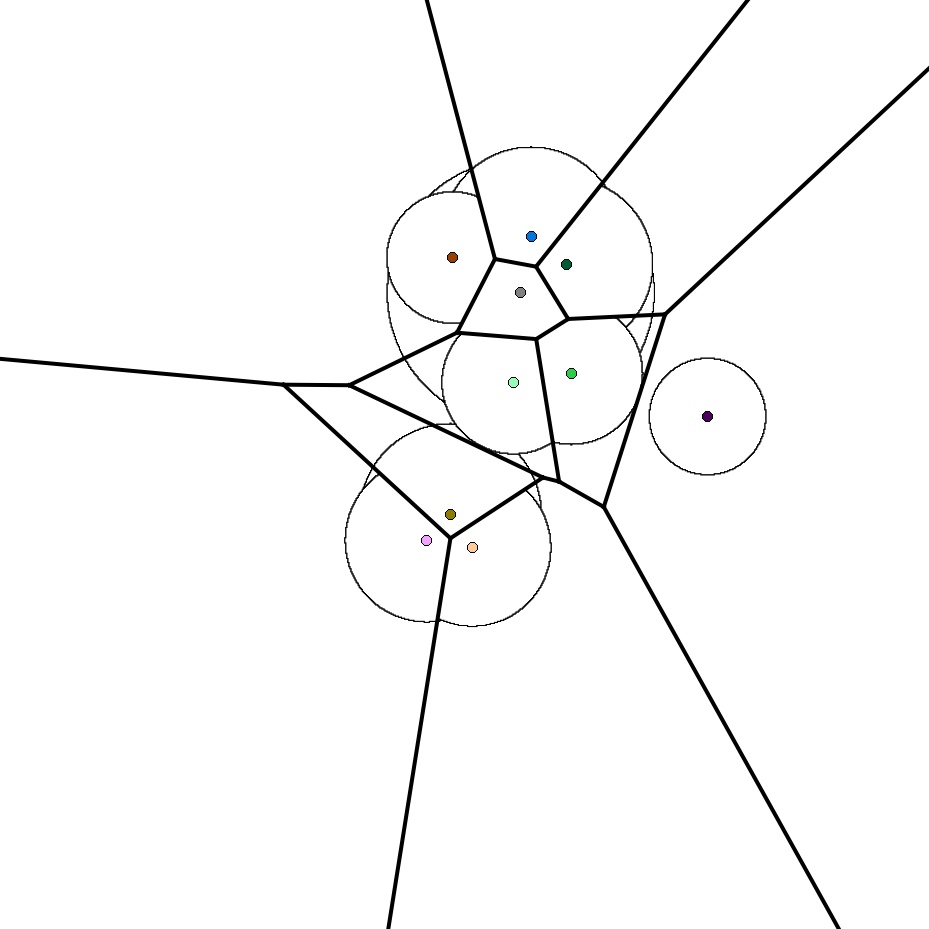}}\\[2pt]
		\fbox{\includegraphics[trim={0 4.7cm 0 2cm},clip,width=0.973\linewidth]{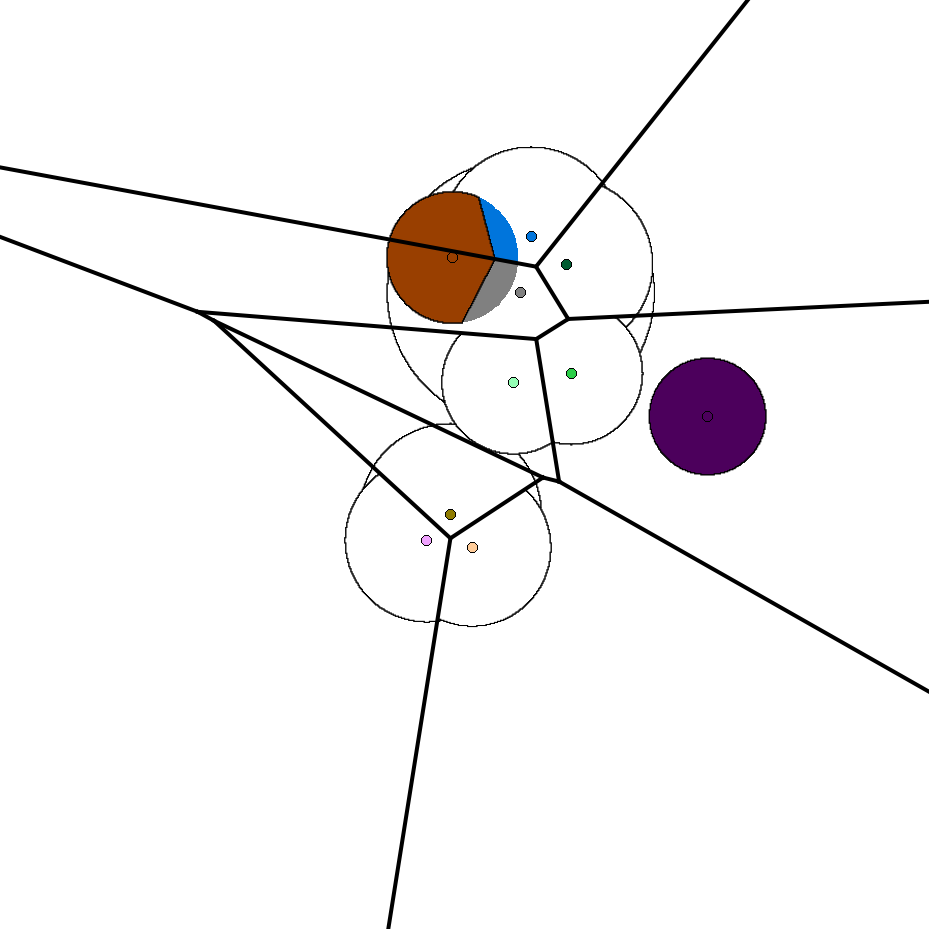}}\\[2pt]
		\fbox{\includegraphics[trim={0 4.7cm 0 2cm},clip,width=0.973\linewidth]{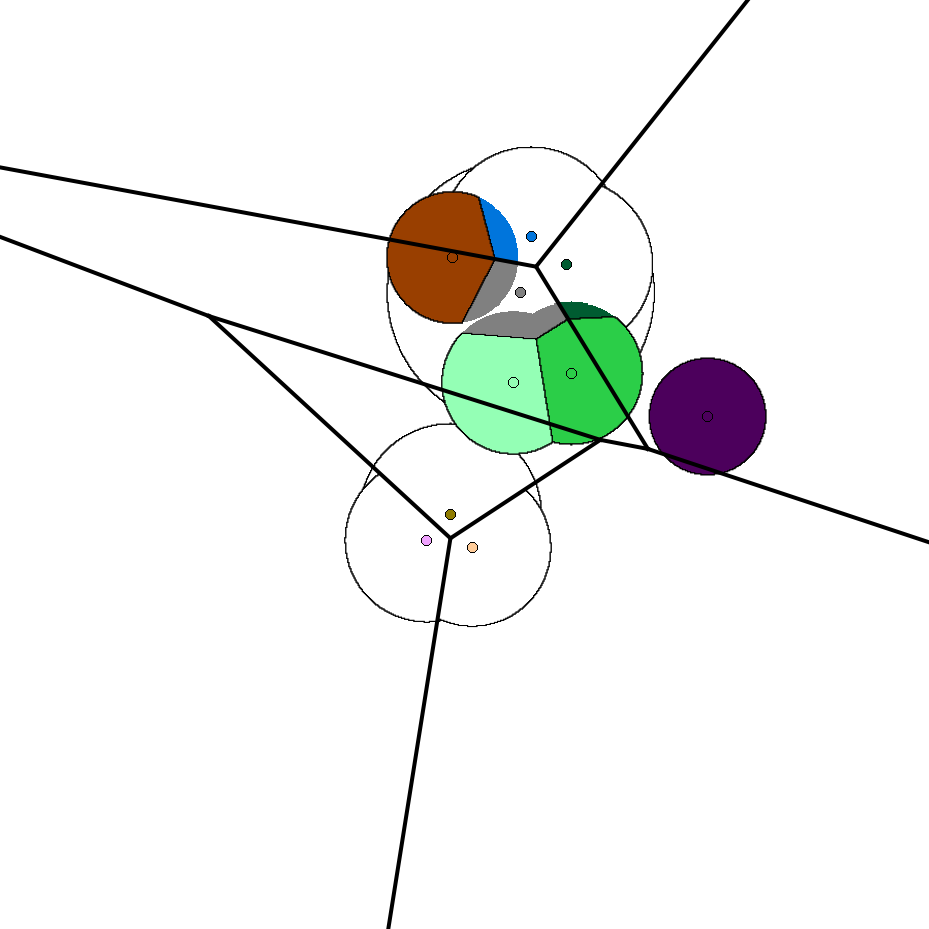}}
	\end{minipage} \begin{minipage}[b]{0.48\linewidth}
		\fbox{\includegraphics[trim={0 4.7cm 0 2cm},clip,width=0.973\linewidth]{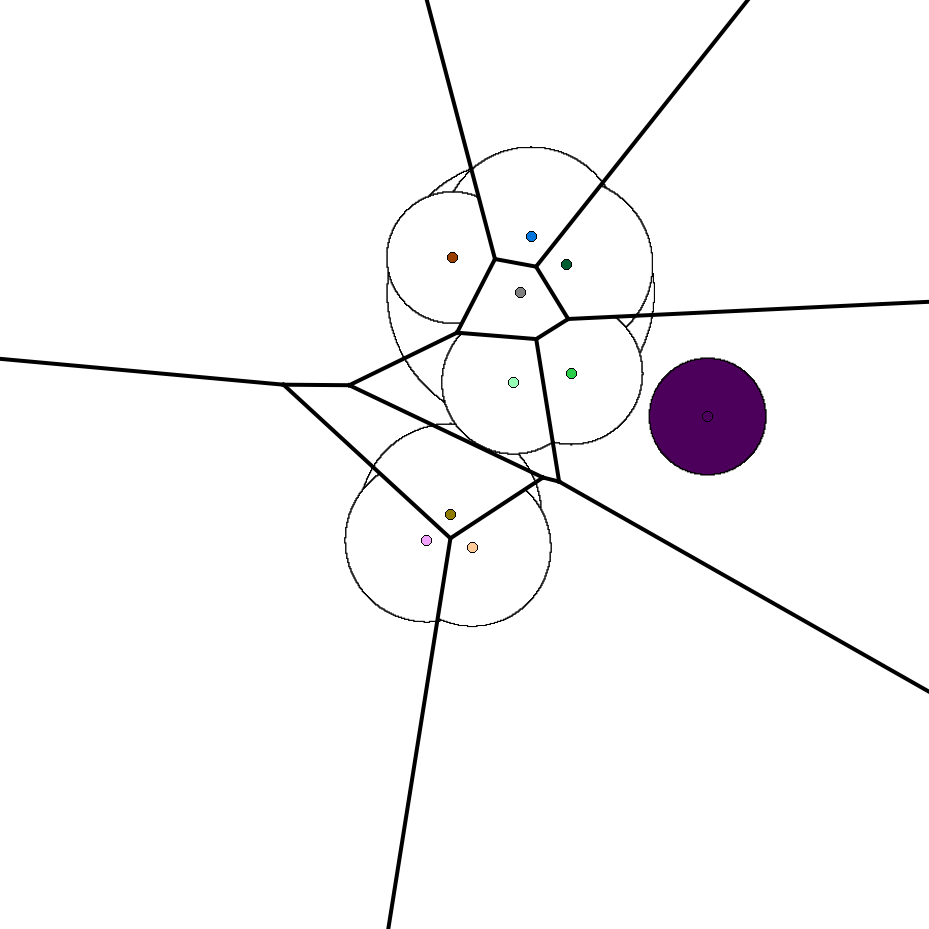}}\\[2pt]
		\fbox{\includegraphics[trim={0 4.7cm 0 2cm},clip,width=0.973\linewidth]{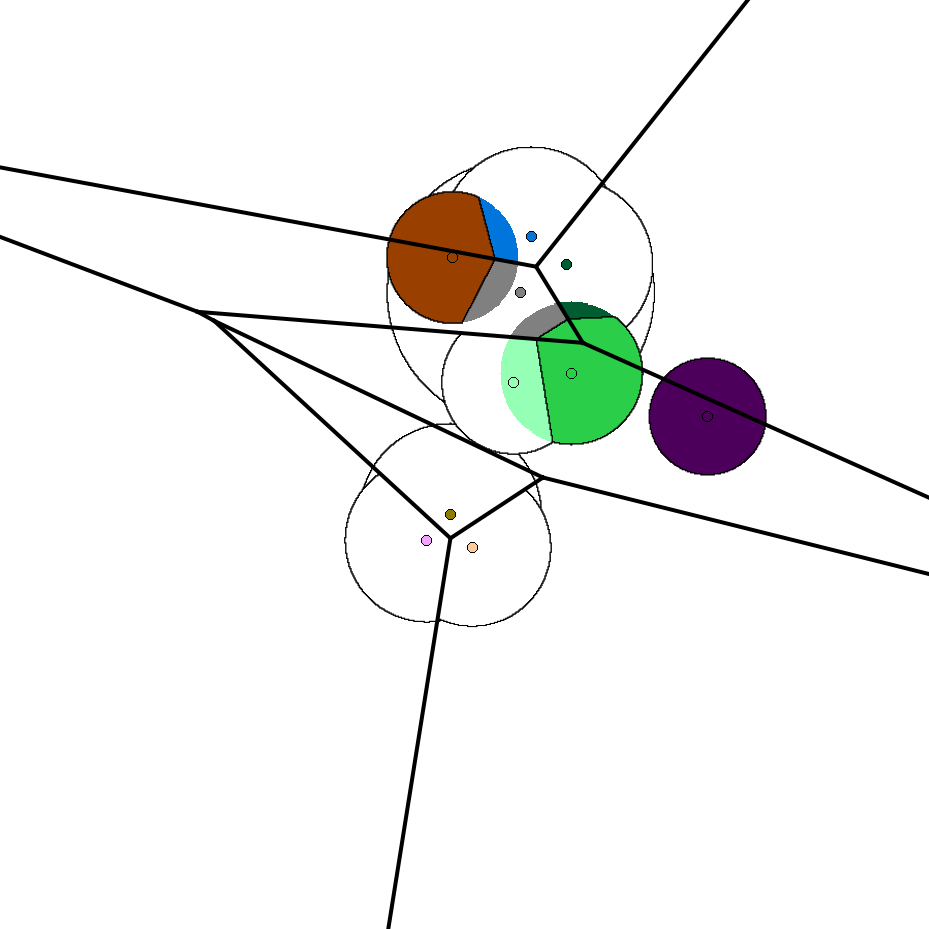}}\\[2pt]
		\fbox{\includegraphics[trim={0 4.7cm 0 2cm},clip,width=0.973\linewidth]{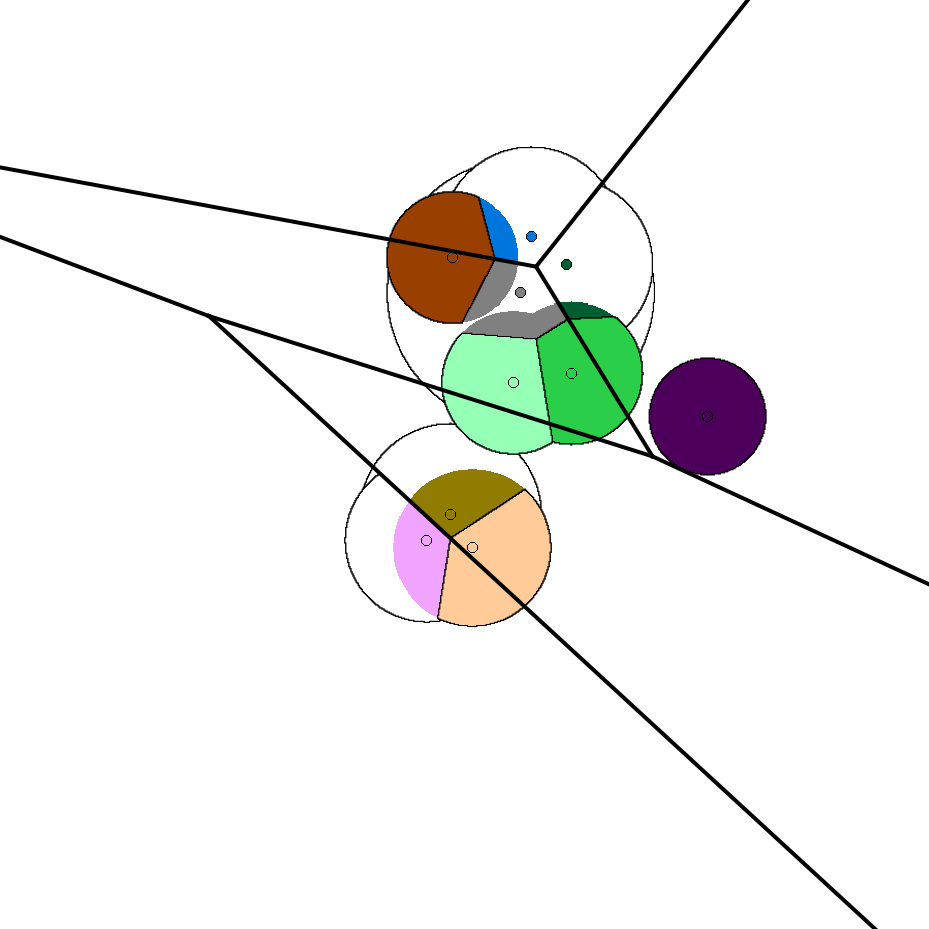}}
	\end{minipage} 
	\caption{Partial diagrams $D_0,\ldots, D_5$ computed in the first five iterations of the algorithm for a set of sites with equal appetites. At each iteration $i$, the edges of the standard Voronoi diagram $V_{i..n}$ of $S_{i..n}$ are overlaid in thick lines. The edges of the stable-matching Voronoi diagram (unknown to the algorithm) are overlaid in thin lines.} \label{smvd:fig:newalgo} 
\end{figure}

\begin{lemma}\label{smvd:lem:corr}
For any $i$ with $0\leq i\leq n$, $D_i$ is coherent with $D^*$.
\end{lemma}

\begin{proof}
We use induction on $i$. The claim is trivial for $i=0$, as no site has a region in $D_0$.
We show that if $D_{i-1}$ is coherent with $D^*$ and $D_i$ is constructed as described, $D_i$ is also coherent. In other words, we show that $\hat{V}_{i..n}$ is coherent with $D^*$.

Let $s$ be an arbitrary site in $S_{i..n}$. We need to show that $\hat{V}_{i..n}(s)\subseteq D^*(s)$.
Let $p$ be an arbitrary point in the interior of $\hat{V}_{i..n}(s)$. We show that $p$ is also an interior point of $D^*(s)$.
First, note that $p$ does not belong in the stable cell of any of $s_1,\ldots,s_{i-1}$, because the regions of these sites are fully contained in $\cup B_{1..i-1}$, and $\hat{V}_{i..n}$ is disjoint from $\cup B_{1..i-1}$ except perhaps along boundaries.

By virtue of being in the interior of $V_{i..n}(s)$, $p$ has $s$ as first choice among the sites in $S_{i..n}$. We show that $s$ also prefers $p$ over \textit{some} of its assigned points in $D^*$, and thus they need to be matched or they would be a blocking pair. We consider two cases:
\begin{itemize}
	\item $s=s_i$. In this case, $\hat{V}_{i..n}$ is a subset of $B^*(s)$, so $s$ prefers $p$ over some of its matched points (those at distance $r^*(s)$).
	\item $s\not= s_i$. In this case, note the following three inequalities:
	\textit{(i)} $d(p,s)<d(p,s_i)$ because $p$ is in the interior of $V_{i..n}(s)$;
	\textit{(ii)} $d(p,s_i)< r^*(s_i)$ because $p$ is in the interior of $B^*(s_i)$;
	\textit{(iii)} $r^*(s_i)\leq r^*(s)$ because $s$ appears after $s_i$ in the proper ordering. Chaining all three, we get that $d(p,s)< r^*(s)$, i.e., $p$ is inside the bounding disk of $s$. Thus, $s$ prefers $p$ to some of its matched points (those at distance $r^*(s)$).
\end{itemize}
\end{proof}

\begin{corollary}
The diagrams $D_n$ and $D^*$ are the same.
\end{corollary}

\begin{proof}
The domain of $D_n$ is $\cup B_{1..n}$ by construction. The domain of $D^*$ is also $\cup B_{1..n}$ by Lemma~\ref{smvd:lem:bounding}. By Lemma~\ref{smvd:lem:corr}, they are coherent, and so it must be that $D_n(s)=D^*(s)$.
\end{proof}

\paragraph{Finding the next bounding disk.}
The proper ordering $s_1,\ldots,s_n$ cannot be computed in a preprocessing step. Instead, the next site $s_i$ is discovered at each iteration. 
Consider the point where we have computed $D_{i-1}$ and want to construct $D_i$ ($1\leq i\leq n$).
At this point, we have found the ordering up to $s_{i-1}$. Hence, we know which sites are in $S_{i..n}$, but we do not know their ordering yet.
In this step, we need to find a site $s$ in $S_{i..n}$ minimizing $r^*(s)$, and we need to find the radius $r^*(s)$ itself. The site $s$ can then be the next site in the ordering, i.e., we can ``label'' $s$ as $s_i$. If there is a tie for the smallest bounding disk among those sites, then there may be several valid candidates for the next site $s_i$. The algorithm finds any of them and labels it as $s_i$.

To find a site $s$ in $S_{i..n}$ minimizing $r^*(s)$, note the following results.
\begin{lemma}\label{smvd:lem:correctness}
	If $r^*(s)\leq r^*(s')$, every point $p$ in $D^*(s)$ satisfies $d(p,s)\leq d(p,s')$.
\end{lemma}

\begin{proof}
	Suppose, for a contradiction, that $r^*(s)\leq r^*(s')$ and $p$ is a point in $D^*(s)$, but $d(p,s')<d(p,s)$. 
	Clearly, $p$ prefers $s'$ to $s$. We show that $s'$ also prefers $p$ over some of its assigned points in $D^*$, and thus $p$ and $s'$ are a blocking pair.
	
	If we combine the three inequalities $d(p,s')<d(p,s)$, $d(p,s)\leq r^*(s)$ (because $p$ is in $D^*(s)$), and $r^*(s)\leq r^*(s')$, we see that $d(p,s') < r^*(s')$.
	Thus, $s'$ prefers $p$ to the points matched to $s'$ along the boundary of its bounding disk.
\end{proof}

\begin{corollary}\label{smvd:cor:correctness}
	For any $s_i$ in a proper ordering, $D^*(s_i)\subseteq V_{i..n}(s_i)$.
\end{corollary}

\begin{proof}
	According to Lemma~\ref{smvd:lem:correctness}, every point $p$ in $D^*(s_i)$ satisfies $d(p,s_i)\leq d(p,s_j)$ for any other site $s_j$ with $r_j\geq r_i$, and this includes every site in $S_{i+1..n}$.
\end{proof}

Based on Corollary~\ref{smvd:cor:correctness}, the idea for finding a site with the next smallest bounding disk is to compute what would be the stable cell of each site $s$ in $S_{i..n}$ if it were constrained to be a subset of $V_{i..n}(s)$. As we will see, among those stable cells, the one with the smallest bounding disk is correct.

More precisely, for each site $s$ in $S_{i..n}$, let $A_i(s)=A(s)-area(D_{i-1}(s))$ be the \textit{remaining appetite} of $s$ at iteration $i$: the starting appetite $A(s)$ of $s$ minus the area already assigned to $s$ in $D_{i-1}$. We define an \textit{estimate cell} $D^\dagger_i(s)$ for site $s$ at iteration $i$ as follows:
$D^\dagger_i(s)$ is the union of $D_{i-1}(s)$ and the intersection of $V_{i..n}(s)\setminus \cup B_{1..i-1}$ with a disk centered at $s$ such that that intersection has area $A_i(s)$.
Note that if $area(V_{i..n}(s)\setminus \cup B_{1..i-1})<A_i(s)$, no such disk exists. In this case, $D^\dagger_i(s)$ is not well-defined. If it is well-defined, we use $B^\dagger_i(s)$ to refer to its bounding disk (the smallest disk centered at $s$ that contains $D^\dagger_i(s)$), and $r^\dagger_i(s)$ to refer to the radius of $B^\dagger_i(s)$. Otherwise, we define $r^\dagger_i(s)$ as $+\infty$.

\begin{lemma}\label{smvd:lem:correctestimate}
At iteration $i$, for any site $s\in S_{i..n}$, $r^*(s)\leq r^\dagger_i(s)$. In addition, if $r^*(s)$ is minimum among the radii $r^*$ of the sites in $S_{i..n}$, then, $r^*(s)=r^\dagger_i(s)$ and $D^*(s)=D^\dagger_i(s)$.
\end{lemma}

\begin{proof}
	For the first claim, let $s$ be a site in $S_{i..n}$.
	Since $D_{i-1}$ is coherent with $D^*$ (Lemma~\ref{smvd:lem:corr}), the region $D_{i-1}(s)$ is in both $D^*(s)$ and $D^\dagger_i(s)$. The appetite of $s$ that is not accounted for in $D_{i-1}(s)$ is $A_i(s)$, and it must be fulfilled outside the domain of $D_{i-1}$, $\cup B_{1..i-1}$. 
	
	In $D^\dagger_i(s)$, $s$ fulfills the rest of its appetite with the points in $V_{i..n}(s)\setminus \cup B_{1..i-1}$ closest to it.
	Note that all these points have $s$ as first choice among the sites in $S_{i..n}$. Thus, the remaining sites in $S_{i..n}$ cannot ``steal'' those points away from $s$, so $s$ for sure does not need to be matched to points even further than that. In other words, in the worst case for $s$, in $D^*$, $s$ fulfills the rest of its appetite, $A_i(s)$, with those points, and thus $r^*(s)=r^\dagger_i(s)$. However, in $D^*$, $s$ may partly fulfill that appetite with points outside of $V_{i..n}(s)$ (and outside $\cup B_{1..i-1}$, of course) which are even closer. These points do not have $s$ as first choice, but they may end up not being claimed by a closer site.  Hence, it could also be that $r^*(s)< r^\dagger_i(s)$. For instance, see Figure~\ref{smvd:fig:estimateradii}.
	
	For the second claim, if $r^*(s)$ is minimum, we are in the worst case for $s$, because, according to Corollary~\ref{smvd:cor:correctness}, $s$ fulfills the rest of its appetite in $V_{i..n}(s)$ and not outside.
\end{proof}

\begin{figure}[h!]
	\centering
	\includegraphics[width=.95\linewidth]{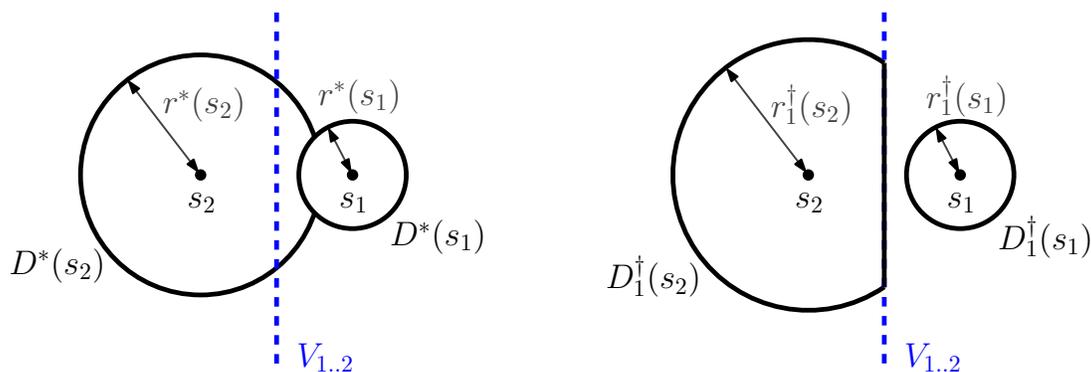}
	\caption{An instance of two sites with different appetites. The left side shows the regions of the sought diagram, $D^*$, and the actual radii $r^*$ of the sites. The right shows the estimate cells and estimate radii of the sites at iteration $1$. We can see that $r^*(s_1)=r_1^\dagger(s_1)$ and that $r^*(s_2)<r_1^\dagger(s_2)$.}
	\label{smvd:fig:estimateradii}
\end{figure}

\begin{corollary}\label{smvd:cor:correctestimate}
At iteration $i$, if $s$ has a smallest estimate radius $r^\dagger_i(s)$ among all the sites in $S_{i..n}$, then $s$ has a smallest actual radius $r^*(s)$ in $D^*$ among all the sites in $S_{i..n}$.
\end{corollary}

Corollary~\ref{smvd:cor:correctestimate} gives us a way to find the next site $s_i$ in a proper ordering: compute the estimate radii of all the sites, and choose a sites with a smallest estimate radius. To do this, we need to be able to compute the estimate radii $r^\dagger_i(s)$.
This is the most challenging step in our algorithm.
Indeed, Observation~\ref{smvd:obs:transcentental} speaks to its difficulty.
To circumvent this problem, we encapsulate the computation of each $r^\dagger_i(s)$ in a geometric primitive that can
be approximated numerically in an
algebraic model of computation.
For the sake of the algorithm description,
we assume the existence of a black-box function that allows us to
compute the following geometric primitive.
\begin{definition}[Geometric primitive]\label{smvd:def:prim} Given a convex polygon $P$, 
	a point $s$ in $P$, an appetite $A$, and a set $C$ of disks, 
	return the radius $r$ (if it exists) 
	such that the area of the intersection of $P\setminus C$ 
	and a disk centered at $s$ with radius $r$ equals $A$.
\end{definition}
In the context of our algorithm, the point $s$ is a site in $S_{i..n}$, 
the appetite $A$ is the remaining appetite $A_i(s)$ of $s$, 
the polygon $P$ is the Voronoi cell $V_{i..n}(s)$, 
and the set of disks $C$ is $B_{1..i-1}$.
Note that such a primitive could be approximated 
numerically to arbitrary precision with a binary search like the one described later in Section~\ref{smvd:sec:convex}.

\paragraph{Implementation and runtime analysis.}

\begin{algorithm}[h!]
	\caption{Stable-matching Voronoi diagram algorithm.}
	\label{smvd:alg:smvd}
	\begin{algorithmic}
		\State \textbf{Input:} set $S$ of $n$ sites, and the appetite $A(s)$ of each site $s$.
		\State Initialize $S_{1..n}$ as $S$, $V_{1..n}$ as a standard Voronoi diagram of $S$, $B_{1..0}$ as an empty set of disks, $\cup B_{1..0}$ as an empty union of disks, and $D_0$ as an empty diagram.
		\State For each site $s\in S$, initialize its remaining appetite $A_1(s)=A(s)$. 
		\For{$i=1,\ldots,n$}
		\For{each site $s$ in $S_{i..n}$}
		\parState{Calculate the estimate radius $r^\dagger_i(s)$ and estimate bounding disk $B^\dagger_i(s)$ of $s$ using the primitive from Definition~\ref{smvd:def:prim} with parameters $V_{i..n}(s),s,A_i(s),$ and $B_{1..i-1}$.}
		\EndFor
		\State Let $s$ be a site in $S_{i..n}$ whose estimate radius $r^\dagger_i(s)$ is minimum.
		\State Set $s_i=s$, $r^*(s_i)=r^\dagger(s_i), B^*(s_i)=B^\dagger(s_i)$.
		\State Compute $\hat{B}(s_i)=B^*(s_i)\setminus \cup B_{1..i-1}$.
		\State Compute $\hat{V}_{i..n}$ by partitioning $\hat{B}(s_i)$ according to $V_{i..n}$.
		\State Add $\hat{V}_{i..n}$ to $D_{i-1}$ to obtain $D_i$.
		\For{each site $s'$ in $S_{i..n}$}
		\State Set $A_{i+1}(s')=A_i(s')-area(\hat{V}_{i..n}(s'))$ ($\hat{V}_{i..n}(s')$ might be empty).
		\EndFor
		\State Set $S_{i+1..n}=S_{i..n}\setminus\{s_i\}$ and $B_{1..i}=B_{1.. i-1}\cup\{B^*(s_i)\}$.
		\State Add $B^*(s_i)$ to $\cup B_{1.. i-1}$ to obtain $\cup B_{1..i}$.
		\State Remove $s_i$ from $V_{i..n}$ to obtain $V_{i+1..n}$.
		\EndFor
		\State Return $D_n$.
	\end{algorithmic}
\end{algorithm}

Given the preceding discussion, Algorithm~\ref{smvd:alg:smvd} shows the full pseudocode.
It uses the following data structures:
\begin{itemize}
	\item $V_{i..n}$: the standard Voronoi diagram of $n$ sites has $O(n)$ combinatorial complexity. It can be initially computed in $O(n\log n)$ time (e.g., see~\cite{Aurenhammer:1991,bookAurenhammer}). It can be updated after the removal of a site in $O(n)$ time~\cite{Gowda83}.
	
	\item $\cup B_{1..i}$: the union of $n$ disks also has $O(n)$ combinatorial complexity~\cite{KLPS,SA95}. To compute $\cup B_{1..i}$ from $\cup B_{1..i-1}$, a new disk can be added to the union in $O(n\log n)$ time, e.g., with a typical plane sweep algorithm.
	
	\item $\hat{V}_{i..n}$: since $\cup B_{1..i-1}$ has $O(n)$ complexity, and the boundary of $B^*(s_i)$ can only intersect each edge of $\cup B_{1..i-1}$ twice, $\hat{B}(s_i)$ also has $O(n)$ complexity. Given that both $\hat{B}(s_i)$ and $V_{i..n}$ have $O(n)$ combinatorial complexity, $\hat{V}_{i..n}$ has $O(n^2)$ combinatorial complexity. The diagram $\hat{V}_{i..n}$ can be computed in $O(n^2\log n)$ time, e.g., with a typical plane sweep algorithm.
	
	\item $D_i$: we do not maintain the faces of the diagram $D_i$ explicitly as ordered sequences of edges. Instead, for each site $s$, we simply maintain the region $D_i(s)$ as the (unordered) set of edges $\cup_{1\leq j\leq i}\;edges(\hat{V}_{j..n}(s))$.
	That is, at each iteration $i$, we add to the edge set of each site $s$ the edges bounding the (possibly empty) region of $s$ in $\hat{V}_{i..n}$. Note that after iteration $j$, the set of edges of $s_j$ does not change anymore. Since $\hat{V}_{i..n}$ has $O(n^2)$ complexity for any $i$, the collective size of the these edge sets is $O(n^3)$ throughout the algorithm.
	
\end{itemize}
We wait until the end of the algorithm to construct a proper data structure representing the planar subdivision $D^*$, e.g., a doubly connected edge cell (DCEL) data structure.
We construct it from the sets of edges collected during the algorithm. Let $E(s_i)=\cup_{1\leq j\leq i}\;edges(\hat{V}_{j..n}(s_i))$ be the set of edges for a site $s_i$. 

\begin{lemma}\label{smvd:lem:glue}
	If all the fragments of edges in $E(s_i)$ that overlap with other parts of edges in $E(s_i)$ are removed, then the parts left are precisely the edges of $D^*(s_i)$, perhaps split into multiple parts. 
\end{lemma}

\begin{proof}
	Since $D^*$ and $D_n$ are the same, every edge $e$ of $D^*(s_i)$ appears in $E(s_i)$. However, $e$ may not appear as a single edge. Instead, it may be split into multiple edges or fragments of edges of $E(s_i)$. This may happen when, for some $s_j$ with $j<i$, the boundary of $\hat{B}(s_j)$  intersects $e$. In this case, in $E(s_i)$, the edge $e$ is split in two at the intersection point. This is because the two parts of the edge are found at different iterations of the algorithm. See, e.g., edge $e$ in Figure~\ref{smvd:fig:glue}.
	
	However, in $E(s_i)$ there may also be edges or fragments of edges which do not correspond to edges of $D^*(s_i)$. These are edges or fragments of edges that actually lie in the interior of $D^*(s_i)$, but that are added to $E(s_i)$ because they lie along the boundary of $\hat{B}(s_j)$ for some $s_j$ with $j<i$, which makes the region of $D^*(s_i)$ be split along that boundary. Such edges appear exactly twice in $E(s_i)$: one for each face on each side of the split. See, e.g., the edges that lie in the interior of $D^*(s_4)$ in Figure~\ref{smvd:fig:glue}, and note that they are all colored twice (unlike the actual edges of $D^*(s_i)$).
\end{proof}

\begin{figure}[htb]
	\centering
	\includegraphics[width=.45\linewidth]{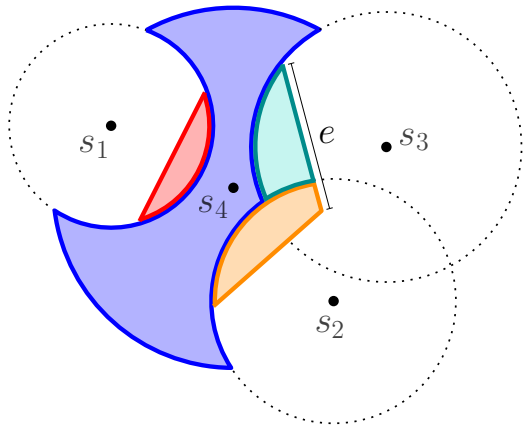}
	\caption{The union of colored regions is the stable cell $D^*(s_4)$ of the site $s_4$. The algorithm finds it divided into four regions, $\hat{V}_{i..4}(s_4)$ for $i=1,2,3,4$, shown in different colors. The bounding disks of $s_1,s_2,$ and $s_3$, are hinted in dotted lines. The edge $e$ of $D^*(s_4)$, which lies along the perpendicular bisector between $s_3$ and $s_4$, is split between $\hat{V}_{2..4}(s_4)$ and $\hat{V}_{3..4}(s_4)$.}
	\label{smvd:fig:glue}
\end{figure}

Given Lemma~\ref{smvd:lem:glue}, we can construct $D^*(s_i)$ from $E(s_i)$ as follows: first, remove all the overlapping fragments of edges in $E(s_i)$. Second, connect the edges with matching endpoints to construct the faces. While doing this, if two straight edges that lie on the same line share an endpoint, merge them into a single edge. Similarly, merge any two curved edges that lie along the same arc and share an endpoint. These are the fragmented edges of $D^*(s_i)$.

Each of these steps can be done with a typical plane sweep algorithm. In more detail, this could be done as follows: sort the endpoints of edges in $E(s_i)$ from left to right. Then, process the edges in the order encountered by a vertical line that sweeps the plane from left to right. 
Maintain all the edges intersecting the sweep line, ordered by height of the intersection (e.g., in a balanced binary search tree). In this way, overlapping edges (for the first step) or edges with a shared endpoint (for the second step) can be found quickly in $O(\log n)$ time. Construct the faces of $D^*(s_i)$ as they are passed by the sweep line. 

Since the sets $E(s_i)$, for $1\leq i\leq n$, have $O(n^3)$ cumulative combinatorial complexity, sorting all the $E(s_i)$ sets can be done in $O(n^3\log n)$ time. The plane sweeps for all the $s_i$ have overall $O(n^3)$ events, each of which can be handled in $O(\log n)$ time. Thus, the algorithm takes $O(n^3\log n)$ time in total.

\begin{theorem}
	\label{smvd:thm:alg}
	The stable-matching Voronoi diagram of a set $S$ of $n$ point sites
	can be computed in the real-RAM model in $O(n^3\log n)$ time plus $O(n^2)$ calls to a 
	geometric primitive that has input complexity $O(n)$.
\end{theorem}

\begin{proof}
For the number of calls to the geometric primitive, note that
there are $n$ iterations, and at each 
iteration we call the geometric primitive $O(n)$ times.
Any given cell of the standard Voronoi diagram $V_{i..n}$ has $O(n)$ edges, and there are $O(n)$ already-matched disks, 
so the input of each call has $O(n)$ size.
Therefore, we make $O(n^2)$ calls to the geometric primitive, each of which
has combinatorial complexity $O(n)$.
	
Besides primitive calls, the bottleneck of each iteration $i$ is computing $\hat{V}_{i..n}$. This can be done in $O(n^2\log n)$ time, for a total of $O(n^3\log n)$ time over all the iterations. 
The final step of reconstructing $D^*$ can also can be done in $O(n^3\log n)$.
\end{proof}

\subsection{Geometric Primitive for Polygonal Convex Distance Functions}\label{smvd:sec:convex}
In this section, we show how to implement the geometric primitive \emph{exactly} for convex distance functions induced by convex polygons. The use of this class of metrics for Voronoi Diagrams was introduced in~\cite{Chew1985}, and studied further, e.g., in~\cite{Ma00}. Intuitively, the polygonal convex distance function, $d_S(a,b)$, induced by a convex polygon $S$, is the factor by which we need to scale $S$, when $S$ is centered at $a$, to reach $b$. Solving the primitive exactly for such metrics is interesting for two reasons. First, this class of distance functions includes many commonly used metrics such as the $L_1$ (Manhattan) and $L_\infty$ (Chebyshev) distances. Second, a convex distance function induced by a regular polygon with a large number of sides can be used to approximate Euclidean distance.

Formally, the distance $d_S(a,b)$ is defined as follows: let $S$ be a convex polygon in $\mathbb{R}^2$ that contains the origin. Then, to compute the distance $d_S(a,b)$ from a point $a$ to a point $b$, we translate $S$ by vector $a$ so that $a$ is at the same place inside $S$ as the origin was. Let $p$ be the point at the intersection of $S$ with the ray starting at $a$ in the direction of $b$. Then, $d_S(a,b)=d(a,b)/d(a,p)$ (where $d(\cdot,\cdot)$ is Euclidean distance).

Convex distance functions satisfy triangle inequality, but they may not be symmetric. Symmetry ($d_S(p,q)=d_S(q,p)$) holds if and only if $S$ is symmetric with respect to the origin~\cite{Ma00}. In this section, we assume that $S$ is symmetric with respect to the origin. Another significant difference with Euclidean distance is that the bisector of two points may contain 2-dimensional regions. This happens when the line through the two points is parallel to a side of $S$~\cite{Ma00}. We assume that such degenerates cases do not happen.\footnote{Alternatively, we may redefine the bisector to go along the clockwise-most boundary of the two-dimensional region, as in~\cite{Chew1985}.}

\begin{figure}[htb]
	\centering
	\includegraphics[width=.4\linewidth]{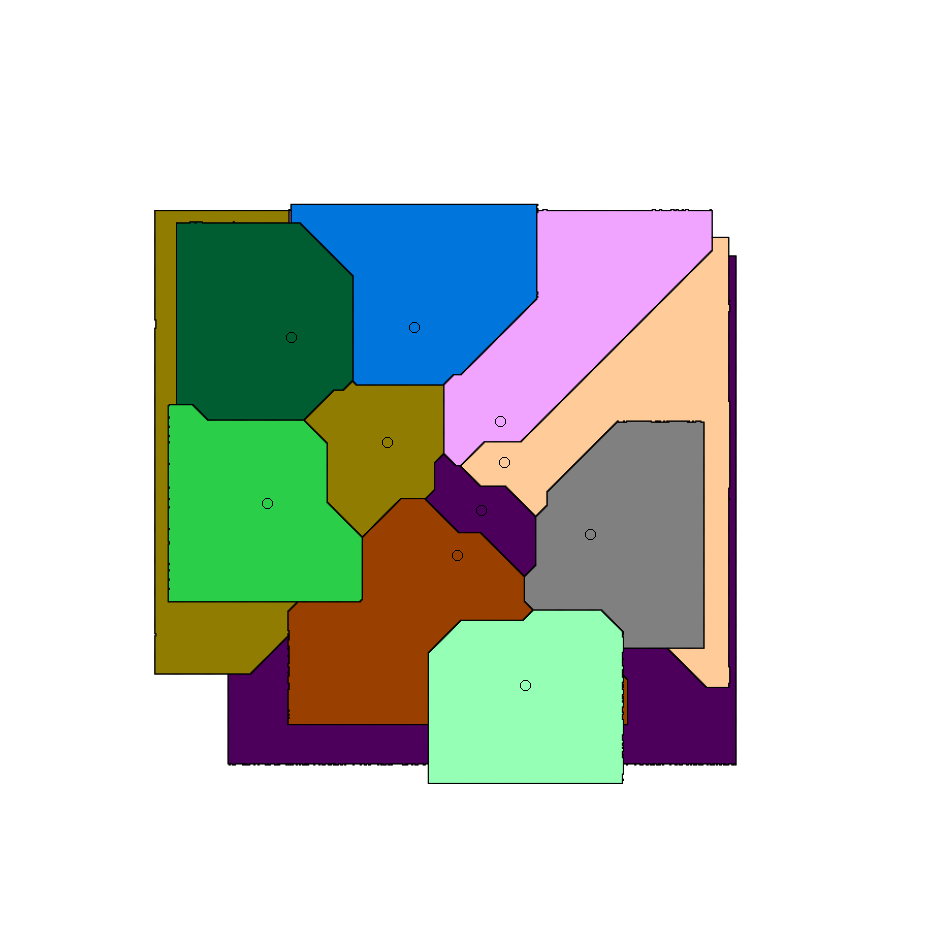}
	\includegraphics[width=.4\linewidth]{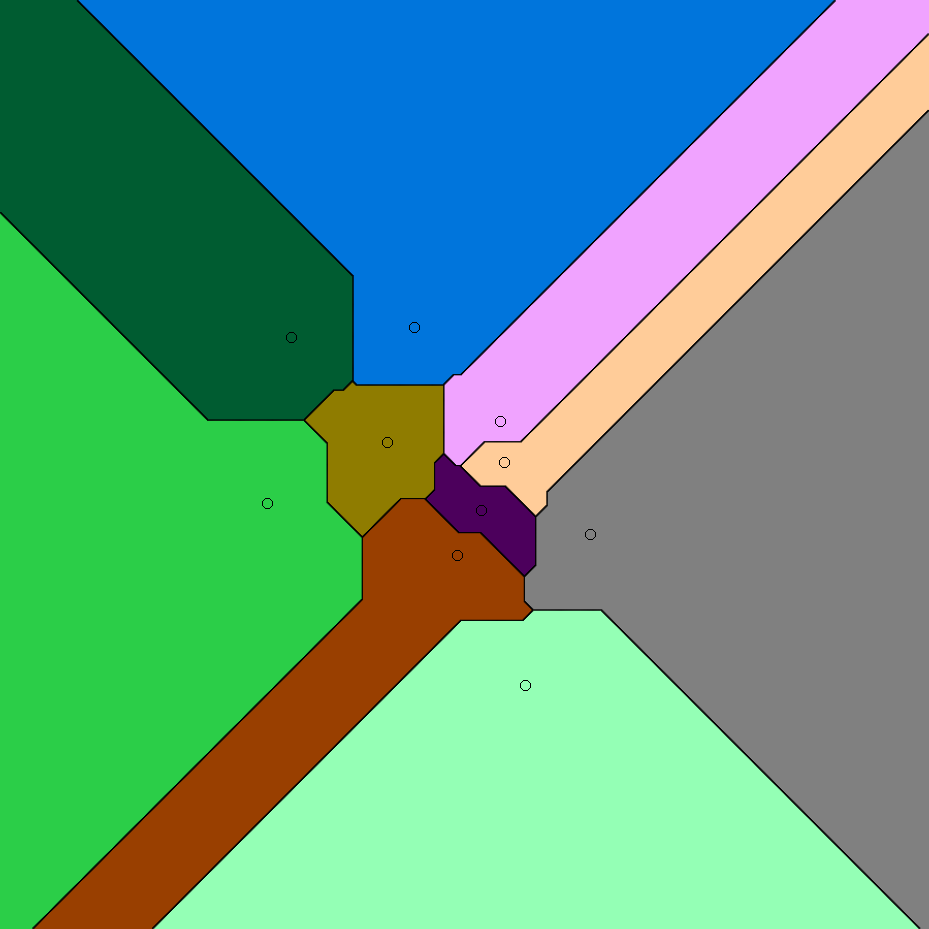}\hspace*{2em}
	\caption{Stable-matching Voronoi diagram (left) and standard Voronoi diagram (clipped by a square) (right) for the convex distance function induced by a square centered at the origin, which corresponds to the $L_\infty$ metric.}
	\label{smvd:fig:chessvoronoi}
\end{figure}

The discussion from Section~\ref{smvd:sec:geo} applies to diagrams based on polygonal convex distance functions. However, in this setting, all the edges are straight. Recall that, in the Euclidean distance setting, straight edges lie along perpendicular bisectors, while curved edges lie along the boundaries of bounding disks. This is still the case here, but bounding disks are constituted of straight edges.
In this context, disks are called balls. A \emph{ball} is a (closed) region bounded by a translated copy of $S$ scaled by some factor.
Therefore, straight and curved edges should now be referred to as \emph{bisector edges} and \emph{bounding ball edges}, respectively. With this distinction, the results in that section also apply.
Likewise, the algorithm applies as well. However, note that the notion of radius is not well defined for convex distance functions, as they grow at different rates in different directions. Therefore, instead of talking about the radii of the bounding disks, we should talk about the scaling factor of the bounding balls.
Most importantly, the fact that there are no curved edges allows us to compute the diagram exactly in an algebraic model of computation. This is the focus of this section.

We need to reformulate the geometric primitive for the case of convex distance functions. Recall that the polygon $P$ in the primitive should correspond to a Voronoi cell, which is the reason why $P$ is assumed to be convex in the primitive. However, Voronoi cells may not be convex for convex distance functions (see Figure~\ref{smvd:fig:chessvoronoi}). Instead, Voronoi cells of polygonal convex distance functions are star-shaped, with the site in the kernel~\cite{Ma00}. Thus, $P$ will now be a star-shaped polygon. For simplicity, we translate the site $s$ to the origin. Finally, we express the solution as the scaling factor of the wanted ball rather than its radius.

\begin{definition}[Geometric primitive for polygonal convex distance functions]\label{smvd:def:prim2} Given a convex distance function induced by a polygon $S$ symmetric with respect to the origin, a star-shaped polygon $P$ with the origin in the kernel, an appetite $A$, and a set $C$ of balls, 
	return the scaling factor $r$ (if it exists) 
	such that $A$ equals the area of the intersection of $P\setminus C$ 
	and $S$ scaled by $r$.
\end{definition}

\paragraph{The algorithm.}
\begin{enumerate}
\item The algorithm begins by computing $P\setminus C$ (which is a polygonal shape that can be concave, have holes, and be disconnected). Then, we triangulate $P\setminus C$ into a triangulation, $T_1$. For each triangle in $T_1$ whose interior intersects one of the \emph{spokes} of $S$ (a ray starting at the origin and going through a vertex of $S$), we divide the triangle along the spoke and re-triangulate each part. After this, the resulting triangulation, $T_2$, has no triangles intersecting any spoke of $S$ except along the boundaries (see Figure~\ref{smvd:fig:triangulation}).

\begin{figure}[htb]
	\centering
	\includegraphics[width=.7\linewidth]{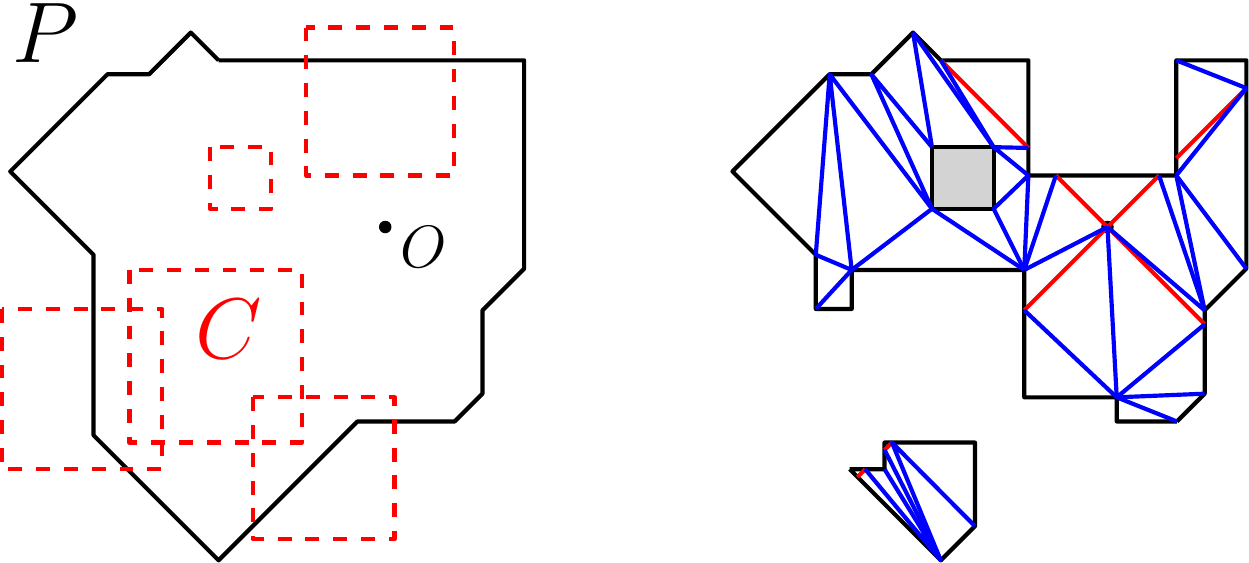}
	\caption{Left: an input to the geometric primitive for the convex distance function induced by a square, where the balls in $C$ are shown in dashed red lines. Right: the corresponding triangulation $T_2$ of $P\setminus C$ where no triangle intersects any spoke of $S$ (shown in red, they are also part of the triangulation).}
	\label{smvd:fig:triangulation}
\end{figure}

\item The next step is to narrow down the range of possible values of $r$. 
We compute, for each vertex $v$ in $T_2$, the distance from the origin $d_S(O,v)$, and sort the vertices from shortest to longest distance. If two or more vertices are at the same distance, we discard all but one, so that we have a sorted list $L$ with only one vertex for each distance. Now, we search for two consecutive vertices $v_1$ and $v_2$ in $L$ such that $d_S(O,v_1)\leq r\leq d_S(O,v_2)$ (or conclude that $r$ does not exist). To find $v_1$ and $v_2$, we can use binary search on the list $L$: for a vertex $v$, we compute the area of the intersection of $P\setminus C$ and a ball centered at the origin passing through $v$ (this can be done by adding the individual contribution of each triangle in $T_2$). By comparing this area to $A$, we discern whether $v$ is too close or too far.

\item It remains to pinpoint $r$ between $d_S(O,v_1)$ and $d_S(O,v_2)$. Let $B_1$ and $B_2$ denote unit balls centered at the origin scaled by $d_S(O,v_1)$ and $d_S(O,v_2)$, respectively, and $B$ the annulus defined by $B_2\setminus B_1$. Note that, because $v_1$ and $v_2$ are consecutive vertices of $L$, the interior of $B$ does not contain any vertex of $T_2$. Conversely, no vertex of $B$ is in the interior of a triangle of $T_2$, because all the vertices of $B$ lie along the spokes of $S$, and no triangle in $T_2$ intersects the spokes of $S$.
As a result, if a triangle in $T_2$ intersects $B$, the intersection is either a triangle or a trapezoid (see Figure~\ref{smvd:fig:btintersections}).
Similarly to Step 1, for each triangle in $T_2$ whose interior is intersected by $B_1$ and/or $B_2$, we divide the triangle along $B_1$ and/or $B_2$ and re-triangulate each part. Figure~\ref{smvd:fig:triangulation2} illustrates the resulting triangulation, $T_3$, where the interior of each triangle is either fully contained in $B$ or disjoint from $B$. Moreover, all the triangles in $B$ have an edge along the boundary of $B_1$ or $B_2$, which we call the base, and a vertex in the boundary of the other (the cuspid).

\begin{figure}[htb]
	\centering
	\includegraphics[width=.575\linewidth]{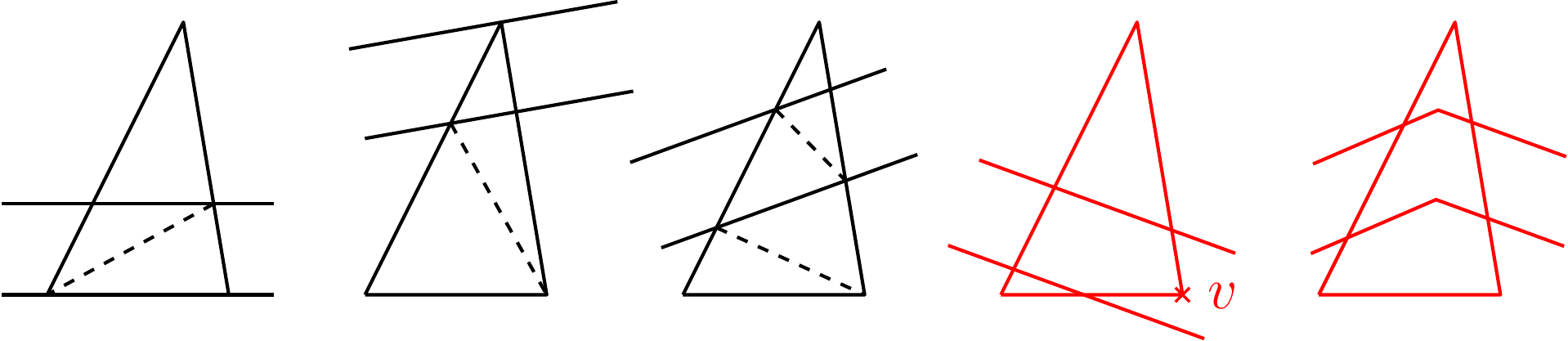}
	\caption{In black: three possible intersections of triangles in $T_2$ and $B$, and the resulting sub-triangulations. In red: two invalid intersections between a triangle in $T_2$ and $B$.} 
	\label{smvd:fig:btintersections}
\end{figure}

\item Finally, we find $r$ as follows. Since $r$ is between $d_S(O,v_1)$ and $d_S(O,v_2)$, triangles outside $B_2$ lie outside the ball with radius $r$. Conversely, all triangles inside $B_1$ are contained in the ball with radius $r$. Let $A'$ be the sum of the areas of all the triangles inside $B_1$. Then, the triangles in $B$ must contribute a total area of $A-A'$. They all have height $h=d_S(O,v_2)-d_S(O,v_1)$. Let $R_1$ and $R_2$ be the sets of triangles in $B$ with the base along  $B_1$ and $B_2$, respectively. We need to find the height $h'$, with $0\leq h'\leq h$, such that $A-A'$ equals the sum of \textit{(i)} the areas of the triangles in $R_1$ from the base to a line parallel to the base at height $h'$, and \textit{(ii)} the areas of the triangles in $R_2$ from the cuspid to a line parallel to the base at height $h-h'$. Given $h'$, we can output $r=d_S(O,v_1)+h'$.

\begin{figure}[htb]
	\centering
	\includegraphics[width=.675\linewidth]{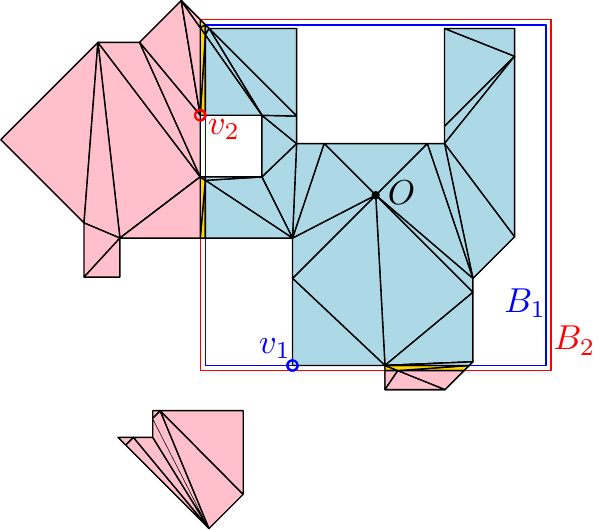}
	\caption{Triangulation $T_3$ of $P\setminus C$ after Step 3 of the algorithm, where no triangle intersects $B$. The triangles of $T_3$ can be classified into those inside $B_1$, inside $B$, and outside $B_2$.}
	\label{smvd:fig:triangulation2}
\end{figure}

In order to find $h'$, we rearrange the triangles to combine them into a trapezoid, as shown in Figure~\ref{smvd:fig:trapezoid}, Left. We rotate the triangles in $R_1$ to align their bases, translate them to put their bases adjacent along a line, and shift their cuspids along a line parallel to the bases to coincide at a single point above the leftmost point of the first base. Doing so does not change their area, and guarantees that triangles do not overlap. We do a similar but flipped transformation to triangles in $R_2$ in order to form the trapezoid. The height $h'$ is the height at which the area of the trapezoid from the base up to that height is $A-A'$, which can be found as the solution to a quadratic equation by using the formula for the area of a trapezoid, as shown in Figure~\ref{smvd:fig:trapezoid}, Right.

\begin{figure}[htb]
	\centering
	\includegraphics[width=.9975\linewidth]{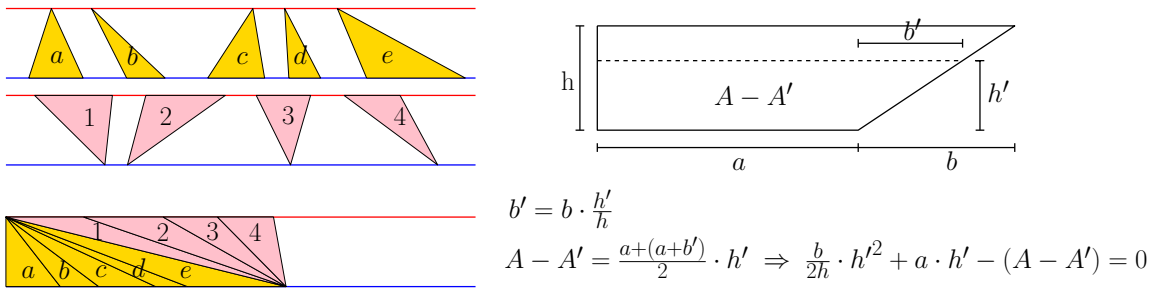}
	\caption{Top left: triangles of $T_3$ inside $B$, rotated and separated into triangles with the base along $B_1$ (top) and $B_2$ (bottom). Bottom left: the triangles rearranged (and transformed) into a trapezoid with the same area. Right: derivation of the quadratic equation for $h'$ from the formula for the area of a trapezoid, for the case where the base is shorter than the top size (i.e., $b$ is positive). Note that $a,b$, and $h$ are known. The alternative case is similar.}
	\label{smvd:fig:trapezoid}
\end{figure}

\end{enumerate}

\paragraph{Running time.} The correctness of the algorithm follows from the simple argument in Step 4. We now consider its runtime analysis. The size of the input to this primitive is $O(|P|+|C|+|S|)$, where $|P|$ and $|S|$ denote the number of edges of the polygons $P$ and $S$, respectively. The polygonal shape $P\setminus C$ has $O(|P|+|C||S|)$ edges, as each ball in $C$ has $|S|$ edges. The corresponding triangulation $T_1$ has $O(|P|+|C||S|)$ triangles. Each spoke of $S$ may intersect every triangle and divide it in two or three, so $T_2$ has $|T_2|=O(|P||S|+|C||S|^2)$ triangles (and vertices). Sorting the vertices of $T_2$ requires $O(|T_2|\log{|T_2|})$ time. The binary search has $O(\log{|T_2|})$ steps, each of which takes time proportional to the number of triangles, $O(|T_2|)$. These steps are the bottleneck, as $T_3$ grows only by a constant factor with respect to $T_2$. Thus, the total runtime of the primitive is $O(|T_2|\log{|T_2|})=O((|P||S|+|C||S|^2)\log{(|P||S|+|C||S|)})$.

In the context of the algorithm, we make calls with to the primitive with $|P|=O(n|S|)$ and $|C|<n$, so we can compute the primitive in $O(n|S|^2\log (n|S|))$ time. When the polygon $S$ has a constant number of faces, the time is $O(n\log n)$. Thus, the entire stable-matching Voronoi diagram for metrics based on these polygons can be computed in $O(n^3\log n)$ total time. This includes the metrics $L_1$ and $L_\infty$.

\section{Conclusions}~\label{smvd:sec:conc}
We have studied stable-matching Voronoi diagrams,
providing characterizations of their combinatorial complexity
and a first discrete algorithm for constructing them.
Stable-matching Voronoi diagrams are a natural generalization of standard Voronoi diagrams to size-constrained regions. This is because standard Voronoi diagrams also have the defining property of stable-matching Voronoi diagrams: stability for preferences based on proximity. Furthermore, both have similar geometric constructions in terms of the lower envelopes of cones.

However, allowing prescribed region sizes comes at the cost of loss of convexity and connectivity; indeed, we have shown that a stable-matching Voronoi diagram may have $O(n^{2+\varepsilon})$ 
faces and edges, for any $\varepsilon>0$. We conjecture that $O(n^2)$ is the right upper bound, matching the lower bound that we have given.

Constructing a stable-matching Voronoi diagram is also 
more computationally challenging than the construction of a standard
Voronoi diagram. In particular, it requires computations that cannot be carried out exactly in an algebraic model of computation. We have given an algorithm which runs in $O(n^3\log n+n^2f(n))$-time, where $f(n)$ is the runtime of a geometric primitive that we defined to encapsulate the computations that cannot be carried out analytically. 
While such primitives cannot be avoided, a step forward from our algorithm would be one that relies only in primitives with constant-sized inputs.

With this work, there are now three approaches for computing stable-matching Voronoi diagrams, each of which requires a different compromise: \textit{(a)} use our algorithm and approximate the geometric primitive numerically; \textit{(b)} replace the Euclidean distance by a polygonal convex distance function induced by a regular polygon with many sides (this approximates a circle, which would correspond to Euclidean distance), and compute the primitive exactly as described in Section~\ref{smvd:sec:convex}; \textit{(c)} discretize the plane into a grid, and use the algorithms of~\cite{EPPSTEIN2017}.

\subparagraph*{Acknowledgements.}
This article reports on work supported by the
DARPA under agreement no.~AFRL FA8750-15-2-0092.
The views expressed are those of the authors and do not reflect the
official policy or position of the Department of Defense
or the U.S.~Government.
Work on this paper by the first author has been supported
in part by BSF Grant~2017684.
This work was also supported in part from NSF grants
1228639, 1526631,
1217322, 1618301, and 1616248.
We would like to thank Nina Amenta for several helpful discussions regarding
the topics of this paper. We also thank the anonymous reviewers for many useful comments.

\bibliography{bibliocopy}

\begin{thebibliography}{10}

\bibitem{SA95}
P.~K. Agarwal and Micha Sharir.
\newblock {Davenport--Schinzel Sequences and Their Geometric Applications}.
\newblock Technical Report DUKE--TR--1995--21, Duke University, Durham, NC,
  USA, 1995.

\bibitem{Aggarwal2009}
Gagan Aggarwal, S.~Muthukrishnan, D\'{a}vid P\'{a}l, and Martin P\'{a}l.
\newblock General auction mechanism for search advertising.
\newblock In {\em 18th Int. Conf. on the World Wide Web (WWW)}, pages 241--250.
  ACM, 2009.
\newblock \href {http://dx.doi.org/10.1145/1526709.1526742}
  {\path{doi:10.1145/1526709.1526742}}.

\bibitem{Aurenhammer:1991}
Franz Aurenhammer.
\newblock {Voronoi diagrams---A survey of a fundamental geometric data
  structure}.
\newblock {\em ACM Computing Surveys}, 23(3):345{--}405, 1991.
\newblock \href {http://dx.doi.org/10.1145/116873.116880}
  {\path{doi:10.1145/116873.116880}}.

\bibitem{aurenhammer1998}
Franz Aurenhammer, Friedrich Hoffmann, and Boris Aronov.
\newblock Minkowski-type theorems and least-squares clustering.
\newblock {\em Algorithmica}, 20(1):61--76, 1998.
\newblock \href {http://dx.doi.org/10.1007/PL00009187}
  {\path{doi:10.1007/PL00009187}}.

\bibitem{bookAurenhammer}
Franz Aurenhammer, Rolf Klein, and Der{-}Tsai Lee.
\newblock {\em Voronoi Diagrams and {Delaunay} Triangulations}.
\newblock World Scientific, 2013.

\bibitem{Bhattacharya2008}
Priyadarshi Bhattacharya and Marina~L. Gavrilova.
\newblock {Roadmap-based path planning -- Using the Voronoi diagram for a
  clearance-based shortest path}.
\newblock {\em IEEE Robotics Automation Magazine}, 15(2):58--66, 2008.
\newblock \href {http://dx.doi.org/10.1109/MRA.2008.921540}
  {\path{doi:10.1109/MRA.2008.921540}}.

\bibitem{Brandt1992}
Jonathan~W. Brandt and V.~Ralph Algazi.
\newblock {Continuous skeleton computation by Voronoi diagram}.
\newblock {\em CVGIP: Image Understanding}, 55(3):329--338, 1992.
\newblock \href {http://dx.doi.org/10.1016/1049-9660(92)90030-7}
  {\path{doi:10.1016/1049-9660(92)90030-7}}.

\bibitem{Chew1985}
L.~Paul Chew and Robert L.~(Scot) Dyrsdale, III.
\newblock Voronoi diagrams based on convex distance functions.
\newblock In {\em Proceedings of the First Annual Symposium on Computational
  Geometry}, SCG '85, pages 235--244, New York, NY, USA, 1985. ACM.
\newblock URL: \url{http://doi.acm.org/10.1145/323233.323264}, \href
  {http://dx.doi.org/10.1145/323233.323264} {\path{doi:10.1145/323233.323264}}.

\bibitem{Cohen-Addad:2018}
Vincent Cohen-Addad, Philip~N. Klein, and Neal~E. Young.
\newblock Balanced centroidal power diagrams for redistricting.
\newblock In {\em Proceedings of the 26th ACM SIGSPATIAL International
  Conference on Advances in Geographic Information Systems}, SIGSPATIAL '18,
  pages 389--396, New York, NY, USA, 2018. ACM.
\newblock URL: \url{http://doi.acm.org/10.1145/3274895.3274979}, \href
  {http://dx.doi.org/10.1145/3274895.3274979}
  {\path{doi:10.1145/3274895.3274979}}.

\bibitem{EPPSTEIN20172short}
David Eppstein, Michael~T. Goodrich, Doruk Korkmaz, and Nil Mamano.
\newblock {Defining equitable geographic districts in road networks via stable
  matching}.
\newblock In {\em 25th ACM SIGSPATIAL Int. Conf. on Advances in Geographic
  Information Systems}, 2017.

\bibitem{EPPSTEIN2017}
David Eppstein, Michael~T. Goodrich, and Nil Mamano.
\newblock {Algorithms for stable matching and clustering in a grid}.
\newblock In {\em 18th Int. Workshop on Combinatorial Image Analysis (IWCIA)},
  volume 10256 of {\em LNCS}, pages 117{--}131. Springer, 2017.
\newblock \href {http://dx.doi.org/10.1007/978-3-319-59108-7_10}
  {\path{doi:10.1007/978-3-319-59108-7_10}}.

\bibitem{Fortune87}
Steven Fortune.
\newblock A sweepline algorithm for {V}oronoi diagrams.
\newblock {\em Algorithmica}, 2(2):153--174, 1987.
\newblock \href {http://dx.doi.org/10.1007/BF01840357}
  {\path{doi:10.1007/BF01840357}}.

\bibitem{gale62}
David Gale and Lloyd~S. Shapley.
\newblock {College admissions and the stability of marriage}.
\newblock {\em The American Mathematical Monthly}, 69(1):9{--}15, 1962.
\newblock \href {http://dx.doi.org/10.2307/2312726}
  {\path{doi:10.2307/2312726}}.

\bibitem{Gowda83}
Ihor~G. Gowda, David~G. Kirkpatrick, Der~Tsai Lee, and Amnon Naamad.
\newblock {Dynamic Voronoi diagrams}.
\newblock {\em IEEE Transactions on Information Theory}, 29(5):724--731,
  September 1983.
\newblock \href {http://dx.doi.org/10.1109/TIT.1983.1056738}
  {\path{doi:10.1109/TIT.1983.1056738}}.

\bibitem{Gusfield:1989:SMP:68392}
Dan Gusfield and Robert~W. Irving.
\newblock {\em The Stable Marriage Problem: Structure and Algorithms}.
\newblock MIT Press, Cambridge, MA, USA, 1989.

\bibitem{hoffman2006}
Christopher Hoffman, Alexander~E. Holroyd, and Yuval Peres.
\newblock {A stable marriage of Poisson and Lebesgue}.
\newblock {\em Annals of Probability}, 34(4):1241{--}1272, 2006.
\newblock \href {http://dx.doi.org/10.1214/009117906000000098}
  {\path{doi:10.1214/009117906000000098}}.

\bibitem{Iwama:2008}
Kazuo Iwama and Shuichi Miyazaki.
\newblock A survey of the stable marriage problem and its variants.
\newblock In {\em IEEE Int. Conf. on Informatics Education and Research for
  Knowledge-Circulating Society (ICKS)}, pages 131--136, 2008.
\newblock URL: \url{http://dx.doi.org/10.1109/ICKS.2008.7}, \href
  {http://dx.doi.org/10.1109/ICKS.2008.7} {\path{doi:10.1109/ICKS.2008.7}}.

\bibitem{KLPS}
Klara Kedem, Ron Livn\'e, J\'anos Pach, and Micha Sharir.
\newblock On the union of {J}ordan regions and collision-free translational
  motion amidst polygonal obstacles.
\newblock {\em Discrete Comput. Geom.}, 1(1):59--71, 1986.
\newblock \href {http://dx.doi.org/10.1007/BF02187683}
  {\path{doi:10.1007/BF02187683}}.

\bibitem{Kise1998}
Koichi Kise, Akinori Sato, and Motoi Iwata.
\newblock {Segmentation of page images using the area Voronoi diagram}.
\newblock {\em Computer Vision and Image Understanding}, 70(3):370--382, 1998.
\newblock \href {http://dx.doi.org/10.1006/cviu.1998.0684}
  {\path{doi:10.1006/cviu.1998.0684}}.

\bibitem{knuth1998art}
Donald~E. Knuth.
\newblock {\em {The Art of Computer Programming, Vol. 3: Sorting and
  Searching}}.
\newblock Addison-Wesley, Reading, MA, 2nd edition, 1998.

\bibitem{Ma00}
Lihong Ma.
\newblock {\em Bisectors and Voronoi Diagrams for Convex Distance Functions}.
\newblock PhD thesis, {FernUniversit{\"a}t Hagen}, 2000.

\bibitem{Meguerdichian2001}
Seapahn Meguerdichian, Farinaz Koushanfar, Gang Qu, and Miodrag Potkonjak.
\newblock Exposure in wireless ad-hoc sensor networks.
\newblock In {\em 7th Int. Conf. on Mobile Computing and Networking (MobiCom)},
  pages 139--150. ACM, 2001.
\newblock \href {http://dx.doi.org/10.1145/381677.381691}
  {\path{doi:10.1145/381677.381691}}.

\bibitem{thematch}
{National Resident Matching Program}, 2017.
\newblock URL: \url{http://www.nrmp.org}.

\bibitem{Petrek2007}
Martin Pet{\v{r}}ek, Pavl{\'\i}na Ko{\v{s}}inov{\'a}, Jaroslav Ko{\v{c}}a, and
  Michal Otyepka.
\newblock {MOLE: A Voronoi diagram-based explorer of molecular channels, pores,
  and tunnels}.
\newblock {\em Structure}, 15(11):1357--1363, 2007.
\newblock \href {http://dx.doi.org/10.1016/j.str.2007.10.007}
  {\path{doi:10.1016/j.str.2007.10.007}}.

\bibitem{Roth89}
Alvin~E. Roth and Marilda Sotomayor.
\newblock {The college admissions problem revisited}.
\newblock {\em Econometrica}, 57(3):559{--}570, 1989.
\newblock \href {http://dx.doi.org/10.2307/1911052}
  {\path{doi:10.2307/1911052}}.

\bibitem{Stojmenovic2006}
Ivan Stojmenovi{\'c}, Anand~Prakash Ruhil, and D.~K. Lobiyal.
\newblock {Voronoi diagram and convex hull based geocasting and routing in
  wireless networks}.
\newblock {\em Wireless Communications and Mobile Computing}, 6(2):247--258,
  2006.
\newblock \href {http://dx.doi.org/10.1002/wcm.384}
  {\path{doi:10.1002/wcm.384}}.

\end{thebibliography}

\appendix

\section{Algorithm Step by Step Illustration}\label{smvd:app:steps}
In this appendix, we illustrate in more detail the incremental construction of the algorithm, by showing several partial diagrams $D_i$ for a set of sites with equal appetites. In these figures, the edges of the standard Voronoi diagram of \textit{all} the sites ($V_{1..n}$) are overlaid in thick black lines, and the edges of the stable-matching Voronoi diagram are overlaid in thin black lines.

\begin{figure}[b!] \centering
\begin{minipage}[b]{0.48\linewidth}
\fbox{\includegraphics[width=0.973\linewidth]{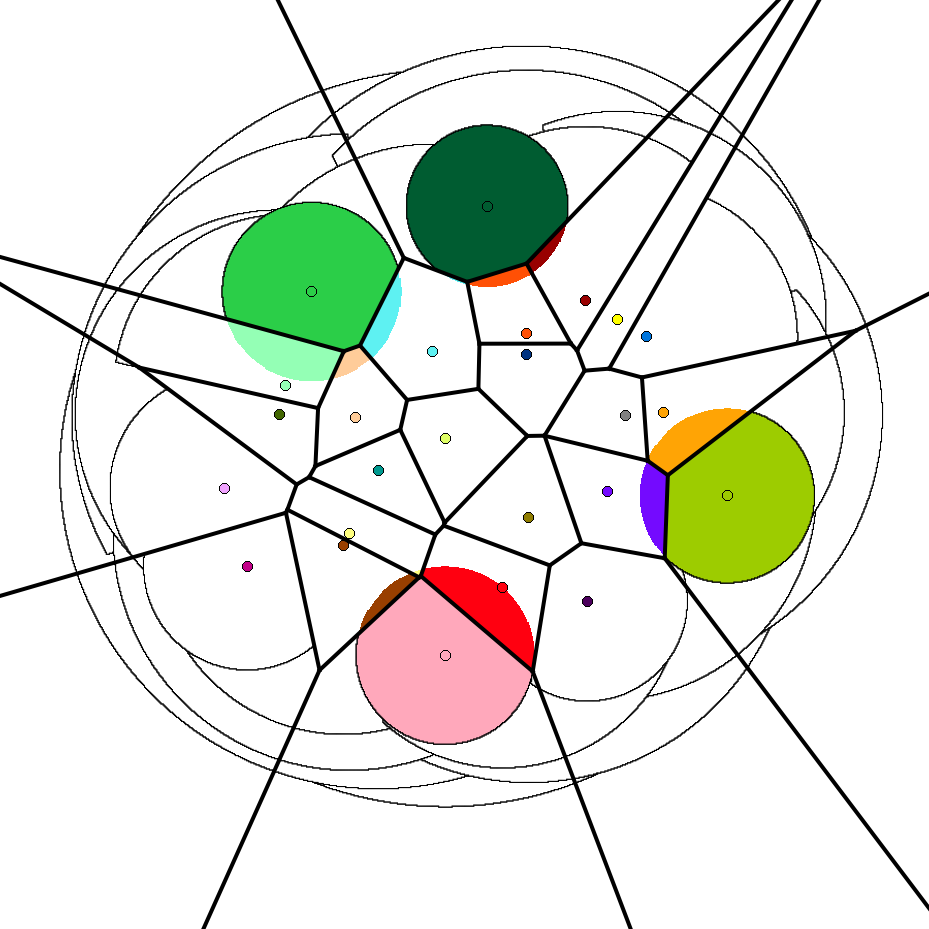}}\\[4pt]
\fbox{\includegraphics[width=0.973\linewidth]{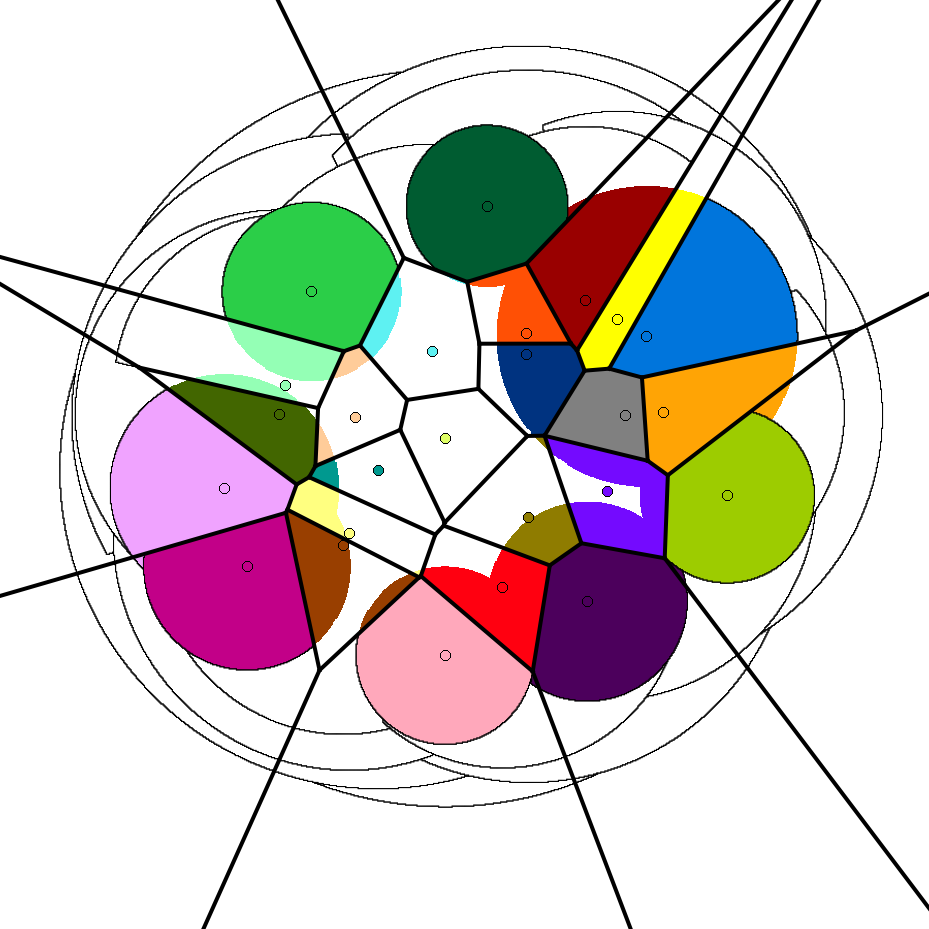}}
\end{minipage} \begin{minipage}[b]{0.48\linewidth}
\fbox{\includegraphics[width=0.973\linewidth]{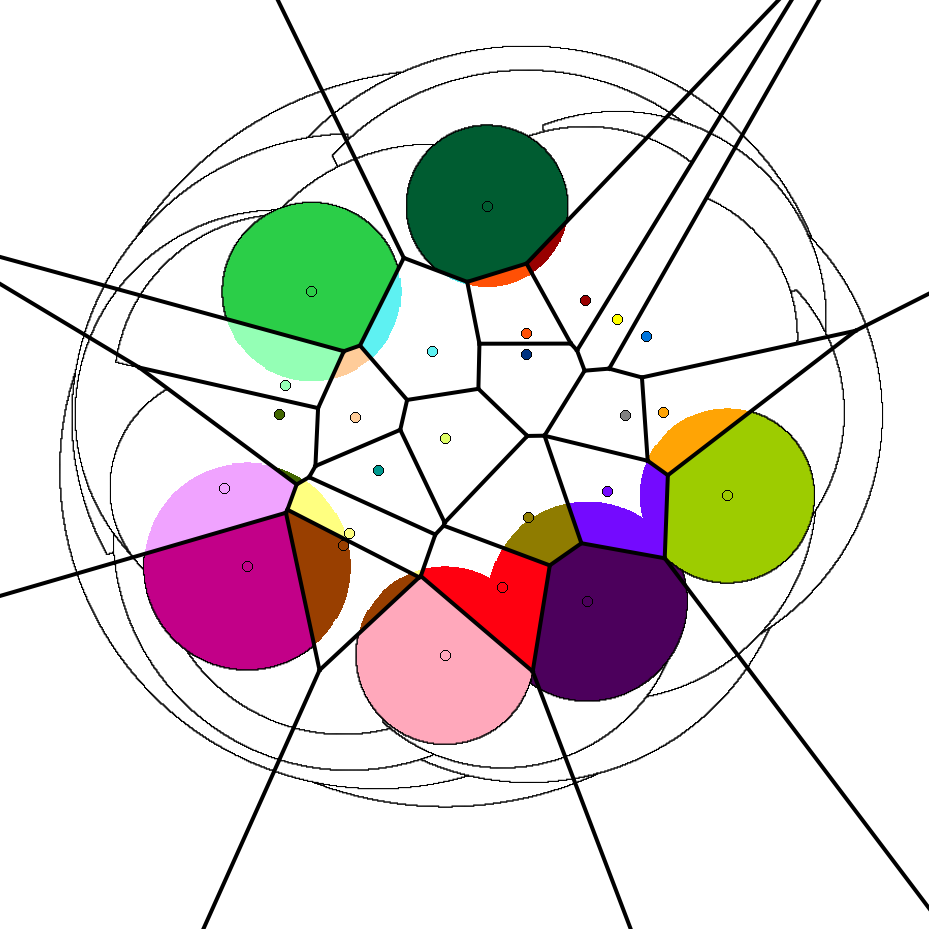}}\\[4pt]
\fbox{\includegraphics[width=0.973\linewidth]{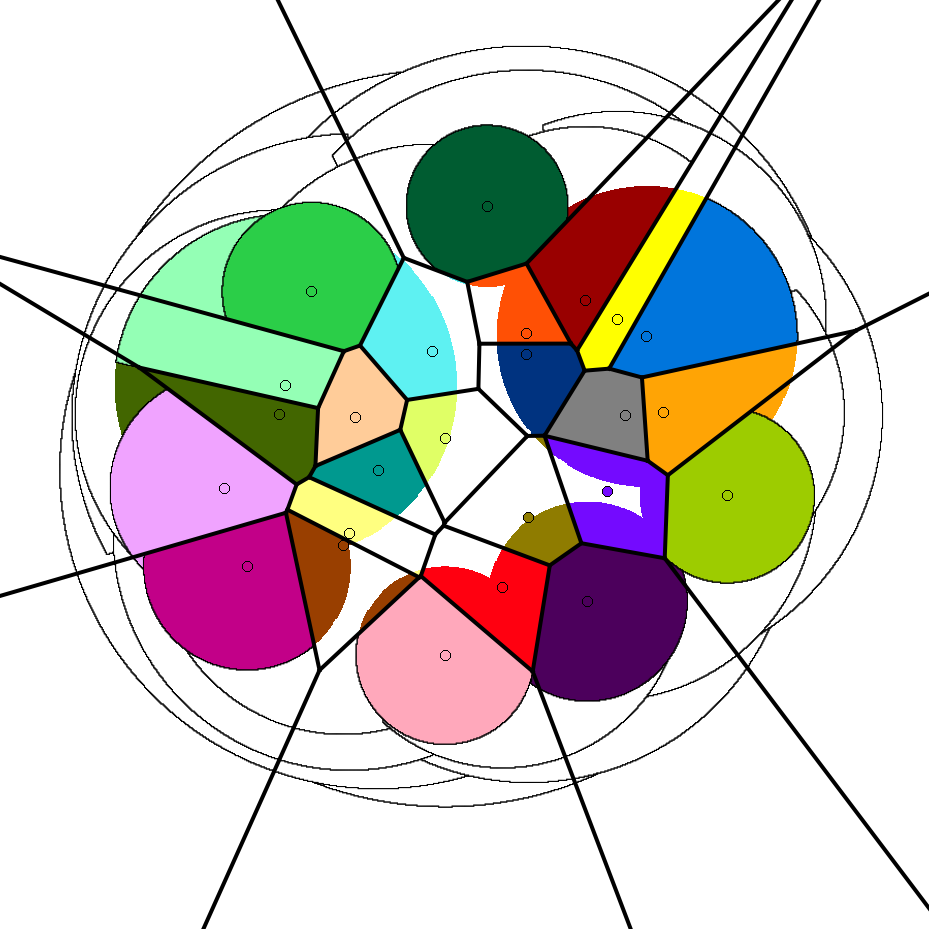}}
\end{minipage} 
\caption{Partial diagram constructed by the Algorithm after 4, 6, 8, and 9 iterations.} \label{smvd:fig:step1} 
\end{figure}

\begin{figure}[ht] \centering
\begin{minipage}[b]{0.48\linewidth}
\fbox{\includegraphics[width=0.973\linewidth]{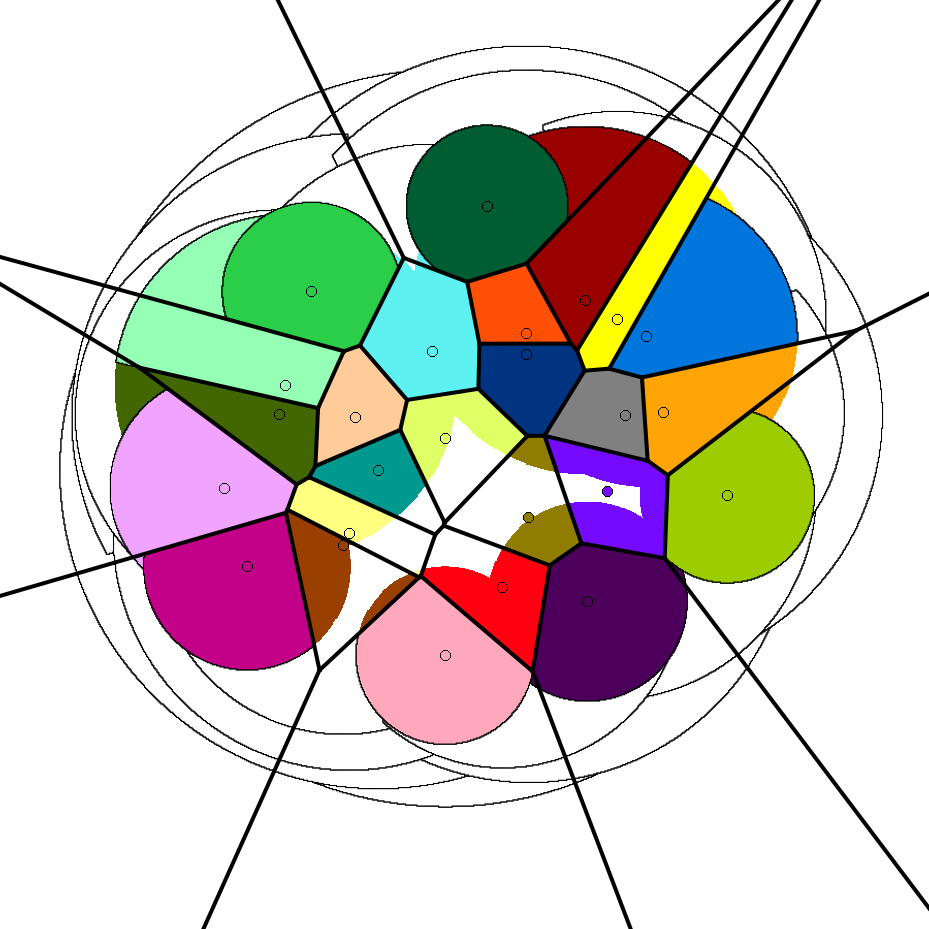}}\\[4pt]
\fbox{\includegraphics[width=0.973\linewidth]{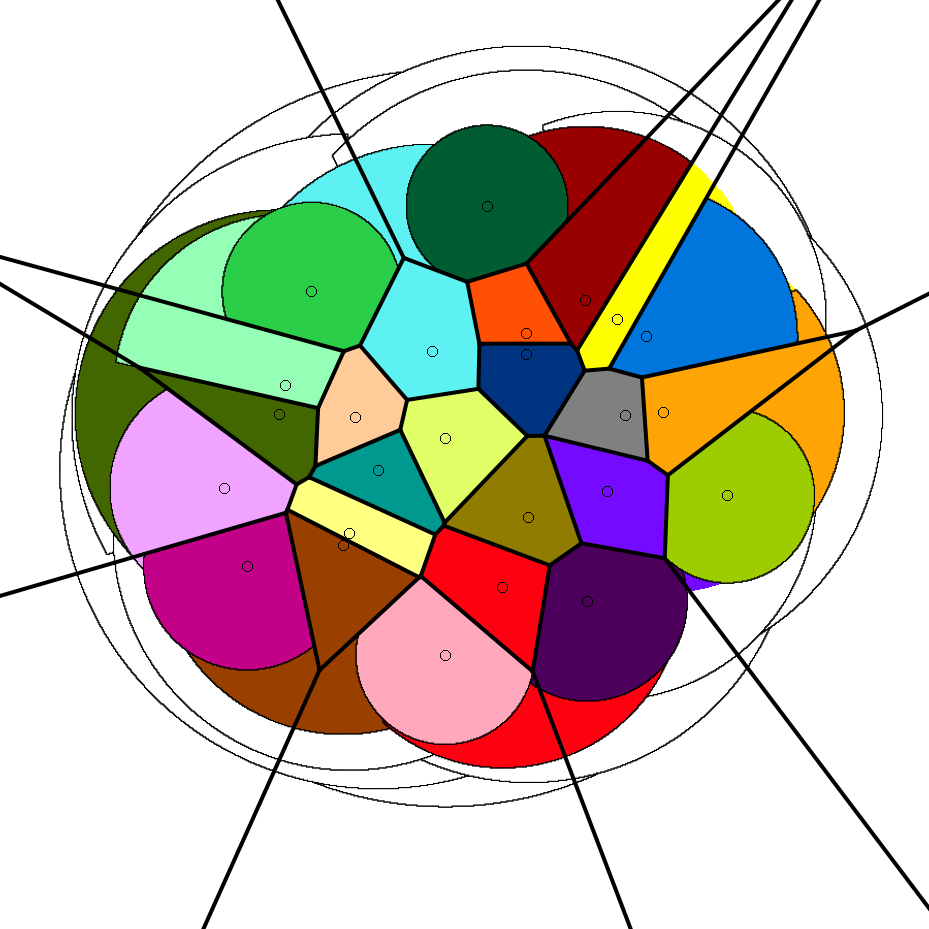}}
\end{minipage} \begin{minipage}[b]{0.48\linewidth}
\fbox{\includegraphics[width=0.973\linewidth]{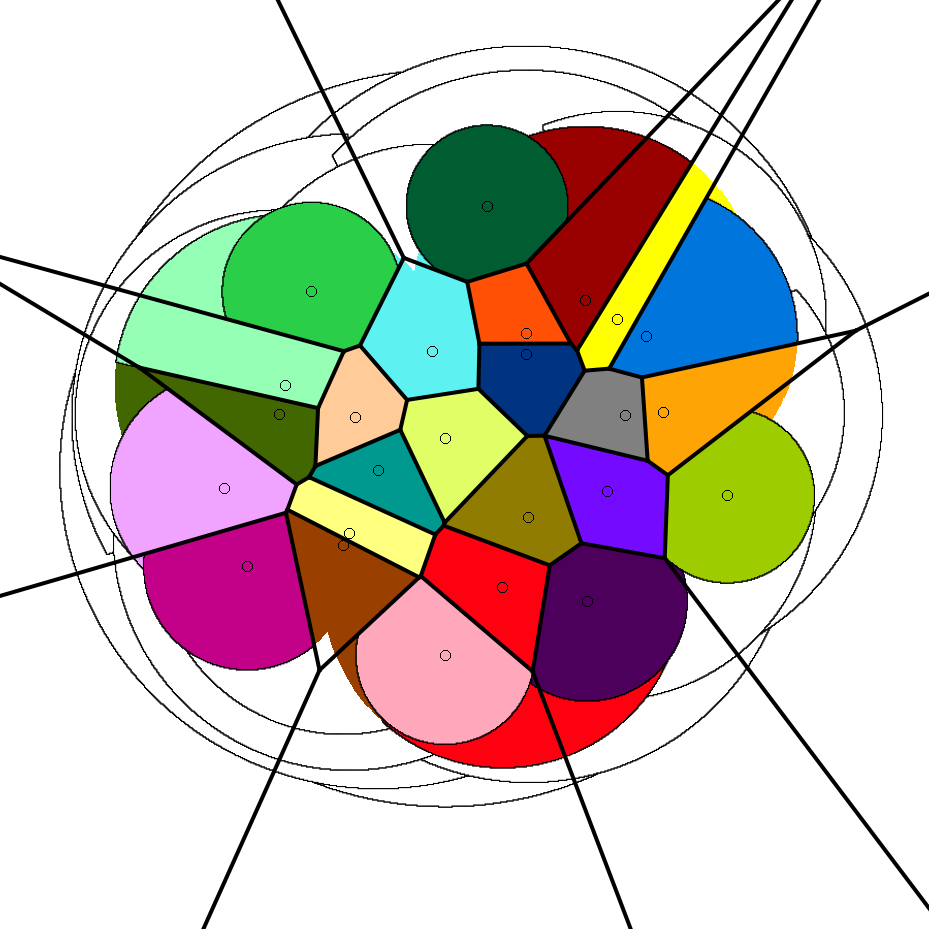}}\\[4pt]
\fbox{\includegraphics[width=0.973\linewidth]{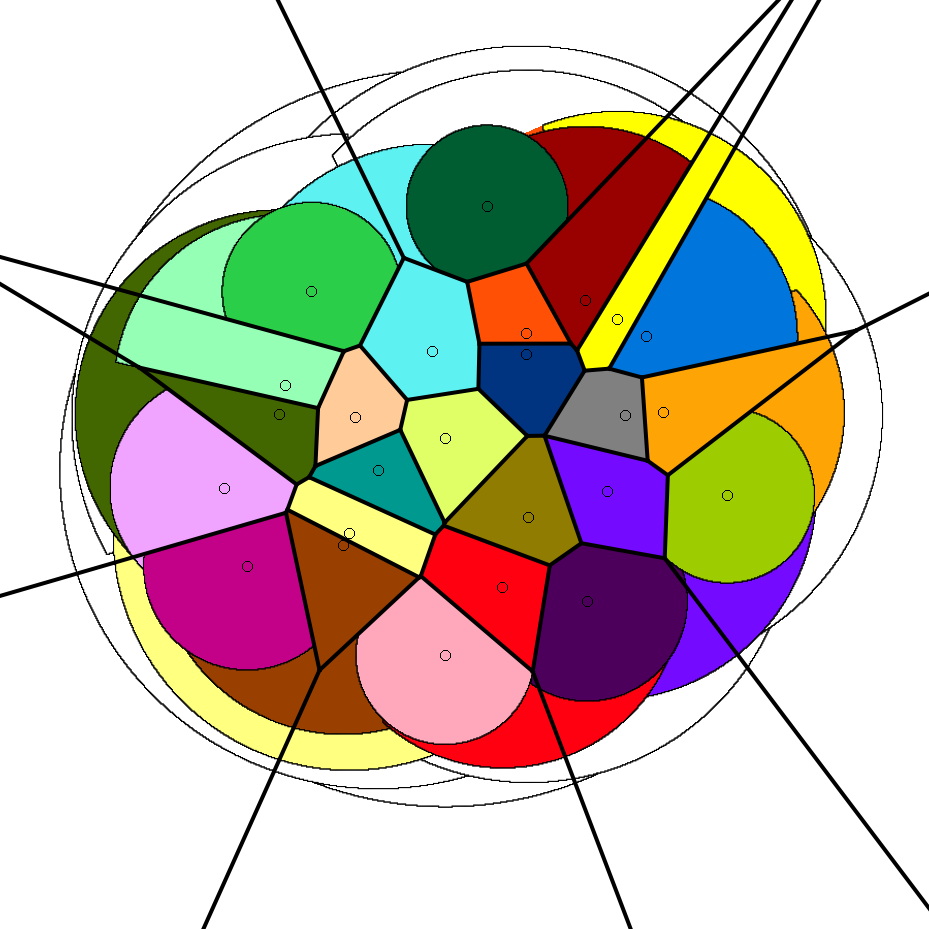}}
\end{minipage} 
\caption{Partial diagram constructed by the Algorithm after 10, 11, 15, and 18 iterations.} \label{smvd:fig:step2} 
\end{figure}

\begin{figure}[ht] \centering
\begin{minipage}[b]{0.48\linewidth}
\fbox{\includegraphics[width=0.973\linewidth]{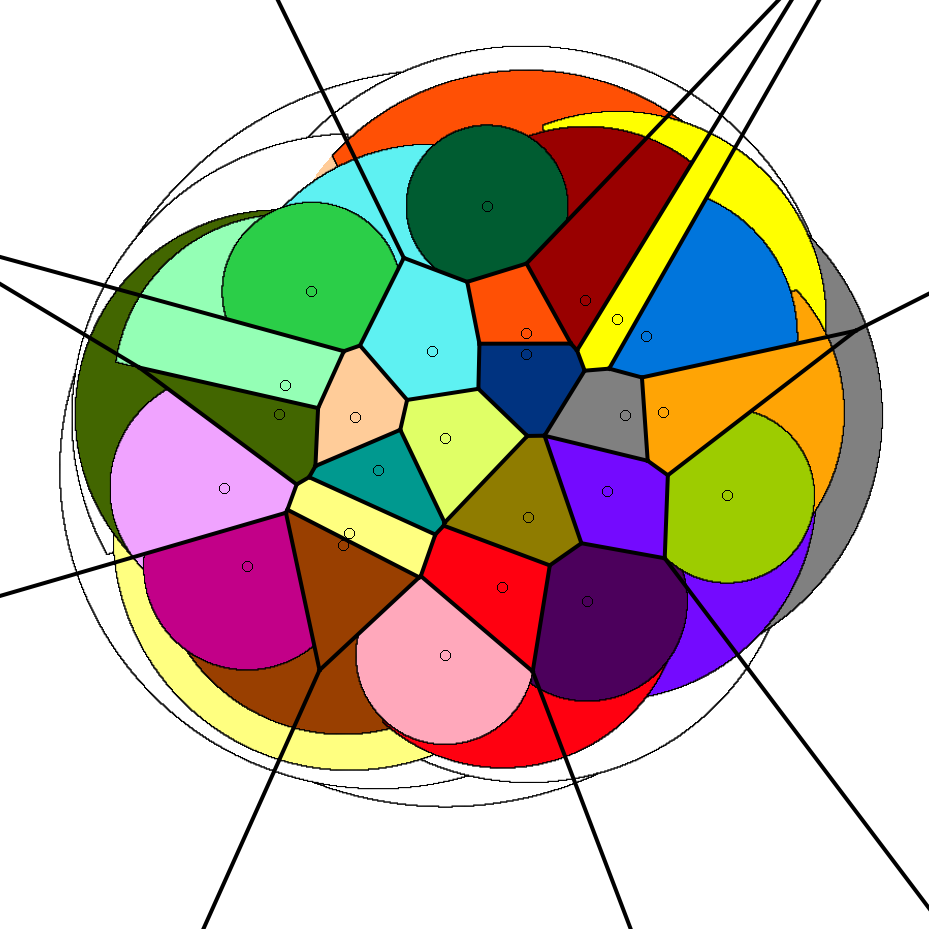}}\\[4pt]
\fbox{\includegraphics[width=0.973\linewidth]{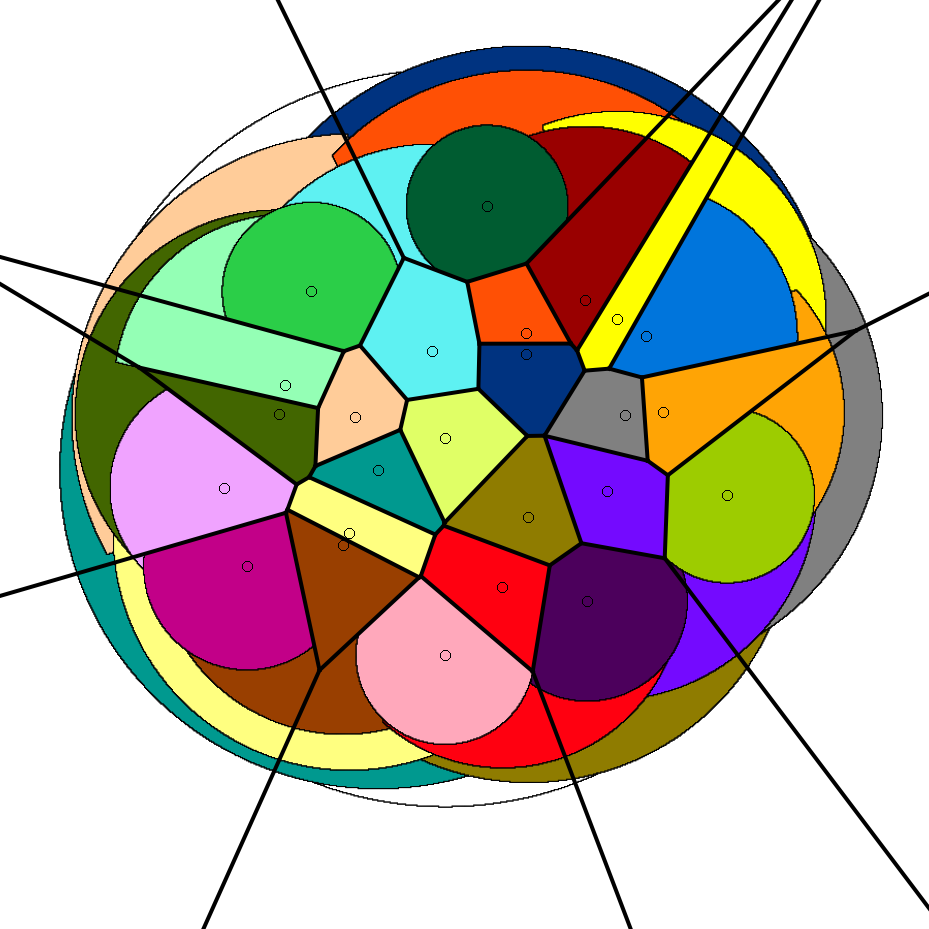}}
\end{minipage} \begin{minipage}[b]{0.48\linewidth}
\fbox{\includegraphics[width=0.973\linewidth]{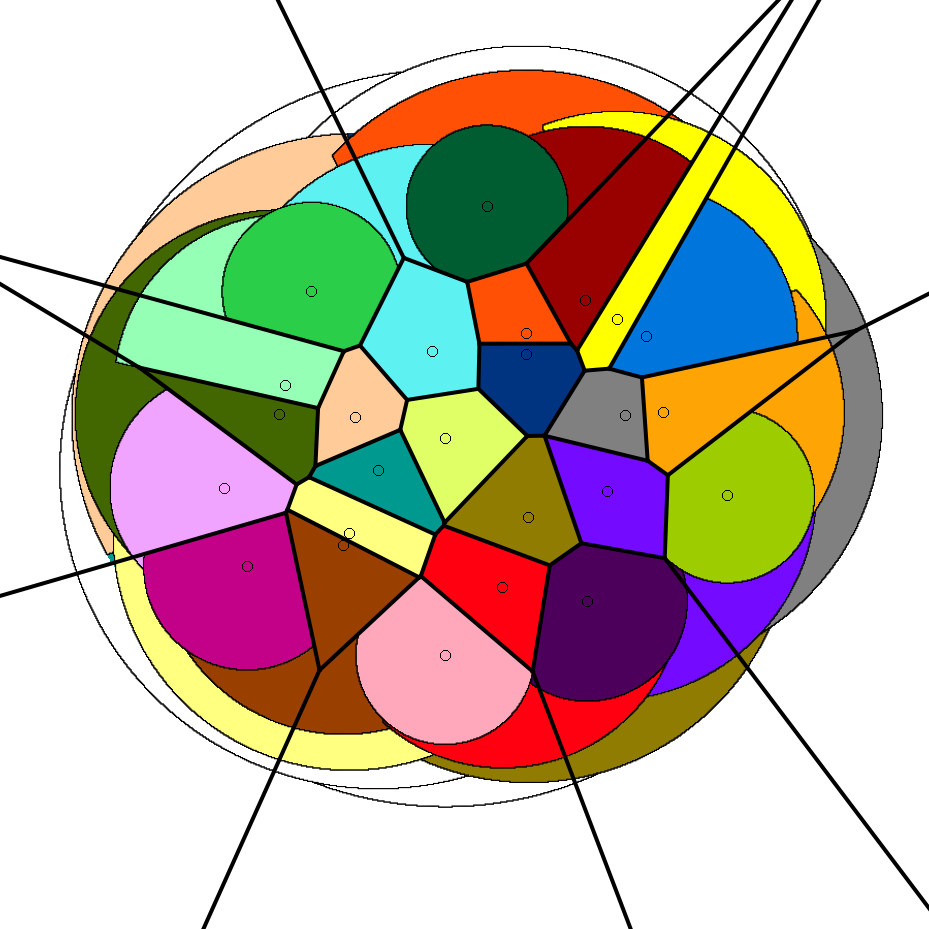}}\\[4pt]
\fbox{\includegraphics[width=0.973\linewidth]{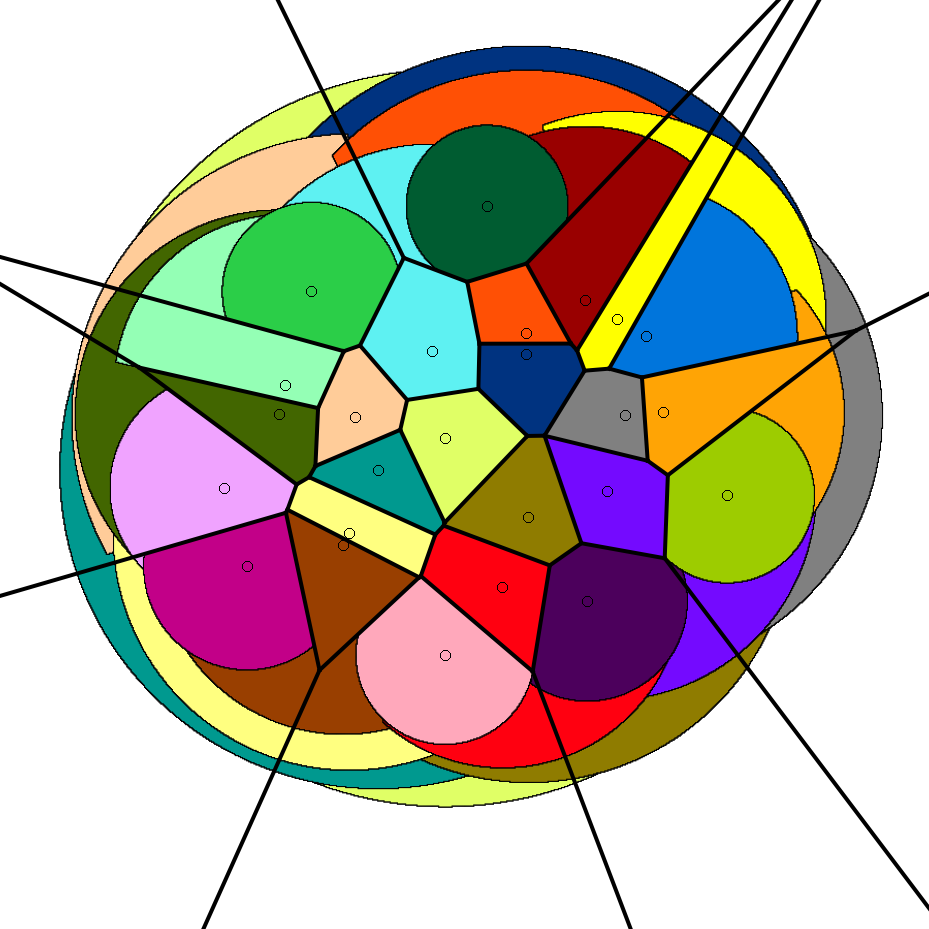}}
\end{minipage} 
\caption{Partial diagram constructed by the Algorithm after 20, 22, 24, and 25 iterations.} \label{smvd:fig:step3}
\end{figure}

\end{document}